\documentclass[ 11pt]{article}
\usepackage{amsmath}
\usepackage{xspace}
\usepackage{amsthm}
\usepackage{epsfig}
\usepackage{boxedminipage}
\usepackage{graphicx}
\usepackage{authblk}

\newcommand{\eat}[1]{}
\newcommand{\eps}{{\varepsilon}}

\newcommand{\opt}{{\tt opt}}

\newcommand{\hsp}{\hspace*{0.2in}}
\newcommand{\proofof}[2]{\emph{Proof of} #1: #2 \qed}

\newtheorem{theorem}{Theorem}[section]

\newtheorem{lemma}[theorem]{Lemma}

\newtheorem{claim}[theorem]{Claim}

\newtheorem{corollary}[theorem]{Corollary}

\renewcommand{\sp}{\hspace*{2 mm}}

\newcommand{\I}{{\cal I}}
\renewcommand{\P}{{\cal P}}
\newcommand{\T}{{\cal T}}

\newcommand{\C}{{\cal C}}

\newcommand{\Seg}{{\tt Seg}}
\renewcommand{\S}{{\cal S}}

\newcommand{\dmc}{{\tt Demand MultiCut}\xspace}

\newcommand{\ts}{{\tilde s}}
\newcommand{\tS}{{\tilde S}}

\newcommand{\wtdfl}{{\tt WtdFlowTime}\xspace}

\newcommand{\State}{{\tt State}}
\newcommand{\St}{{\tt State}}

\newcommand{\cell}{{\tt cell}}

\newcommand{\cells}{{\tt cells}}
\newcommand{\seg}{{\tt seg}}

\newcommand{\redc}{{\tt red}}

\begin{document}

\title{Constant Factor Approximation Algorithm for Weighted Flow Time on a Single Machine in Pseudo-polynomial time}
\date{}

\author{Jatin Batra \qquad Naveen Garg \qquad Amit Kumar}
\affil{Department of Computer Science and Engineering \\ IIT Delhi}


\maketitle
\begin{abstract}
In the weighted flow-time problem on a single machine, we are given a set of $n$ jobs, where each job has a processing requirement $p_j$, release date $r_j$ and weight $w_j$.
The goal is to find a preemptive schedule which minimizes the sum of weighted flow-time of jobs, where the flow-time of a job is the difference between its completion time
and its released date. We give the first pseudo-polynomial time constant approximation algorithm for this problem. The algorithm also extends directly to the problem of minimizing the $\ell_p$ norm of weighted flow-times. The running time of our algorithm is polynomial in $n$, the number
of jobs, and $P$, which is the ratio of the largest to the smallest processing requirement of a job. Our algorithm relies on a novel reduction of this problem to a generalization of
the multi-cut problem on trees, which we call \dmc problem. Even though we do not give a constant factor approximation algorithm for the \dmc problem on trees, we show that the
specific instances of \dmc obtained by reduction from weighted flow-time problem instances have more structure in them, and we are able to employ techniques based
on dynamic programming. Our dynamic programming algorithm relies on showing that there are near optimal solutions which have nice
smoothness properties, and we exploit these properties  to reduce the size of DP table.
\end{abstract}
%
%
\thispagestyle{empty}
\pagebreak
\setcounter{page}{1}
\pagebreak

\section{Introduction}
Scheduling jobs to minimize the average waiting time is one of the most fundamental problems in scheduling theory with numerous applications. We consider the setting where jobs arrive over time (i.e., have release dates), and need to be processed such that the average flow-time is minimized. The flow-time, $F_j$ of a job $j$, is defined as the difference between its completion time, $C_j$,  and release date, $r_j$. It is well known that for the case of single machine, the SRPT policy (Shortest Remaining Processing Time) gives an optimal algorithm for this objective.

In the weighted version of this problem, jobs have weights and we would like to minimize the weighted sum of flow-time of jobs.
However, the problem of minimizing {\em weighted} flow-time (\wtdfl)
turns out to be NP-hard and it has been widely conjectured that there should a constant factor  approximation algorithm (or even PTAS) for it.
  In this paper, we make substantial progress towards this problem by giving the first constant factor approximation algorithm for this problem in pseudo-polynomial time.
More formally, we prove the following result.
\begin{theorem}
\label{thm:main}
There is a constant factor approximation algorithm for \wtdfl where the running time of the algorithm is polynomial in $n$ and $P$. Here, $n$ denotes the number of jobs in the instance, and
$P$  denotes the ratio of the largest to the smallest processing time  of a job  in the instance respectively.
\end{theorem}

We obtain this result by reducing \wtdfl to a generalization of the multi-cut problem on trees, which we call \dmc.
 The \dmc
problem  is a natural generalization of the multi-cut problem where edges have sizes  and costs, and input paths (between terminal pairs)
 have demands. We would like to select a minimum cost
subset of edges such that for every path in the input, the total size of the selected edges in the path is at least the demand of the path. When all demands and sizes are 1, this is the usual
multi-cut problem.
The natural integer program for this problem has the property that all non-zero entries in any column of the constraint matrix are the same.
Such integer programs, called {\em column restricted covering integer programs},
were studied by  Chakrabarty et al.~\cite{ChakrabartyGK10}. They showed that one can get a constant factor approximation algorithm for \dmc
provided one could prove that the integrality gap of the natural LP relaxations for the following two special cases is constant
 -- (i) the version where the constraint matrix  has 0-1 entries only, and (ii) the priority version, where
paths and edges in the tree have priorities (instead of sizes and demands respectively), and  we want to pick minimum cost subset of edges such that
for each path, we pick at least one edge in it of priority which is at least the priority of this path. Although the first problem turns out to be easy,
we do not know how to round the LP
relaxation of the priority version.
This is similar to the situation faced by Bansal and Pruhs~\cite{BansalP14}, where they need to round the priority version of a geometric set cover problem.
They appeal to the notion of shallow cell complexity~\cite{ChanGKS12} to get an $O(\log \log P)$-approximation for this problem.
It turns out the shallow cell complexity of the priority version of \dmc is also unbounded (depends on the number of distinct priorities)~\cite{ChanGKS12},
and so it is unlikely that this approach will yield a constant factor approximation.

However, the specific instances of \dmc produced by our reduction have more structure, namely each node has at most 2 children, each path goes from an ancestor to a descendant, and
the tree has $O(\log (nP))$ depth if we shortcut all degree 2 vertices. We show that one can effectively use dynamic programming techniques for such instances. We show
that there is a near optimal solution which has nice ``smoothness'' properties so that the dynamic programming table can manage with storing small amount of
information.

\subsection{Related Work}
There has been a lot of work on the \wtdfl problem on a single machine, though polynomial time constant factor approximation algorithm has remained elusive.
Bansal and Dhamdhere~\cite{BansalD07} gave an $O(\log W)$-competitive on-line algorithm for this problem, where $W$ is the ratio of the maximum to the minimum weight of a job.
They also gave a semi-online (where the algorithm needs to know the parameters $P$ and $W$ in advance) $O(\log (nP))$-competitive algorithm for \wtdfl, where
$P$ is the ratio of the largest to the smallest processing time of a job.
Chekuri et al.~\cite{ChekuriKZ01} gave a semi-online $O(\log^2 P)$-competitive algorithm.

Recently, Bansal and Pruhs~\cite{BansalP14} made significant  progress towards this problem by giving an $O(\log \log P)$-approximation algorithm. In fact, their result applies to a more general setting where the objective function is $\sum_j f_j(C_j)$, where $f_j(C_j)$ is any monotone function of the completion time $C_j$ of job $j$.
Their work, along with a constant factor approximation for the generalized caching problem~\cite{Bar-NoyBFNS01}, implies a constant factor approximation algorithm for this setting when all release dates are 0.
Chekuri and Khanna~\cite{ChekuriK02} gave a quasi-PTAS for this problem, where the running time was $O(n^{O_\epsilon(\log W \log P)})$. In the special case of stretch metric, where $w_j = 1/p_j$,
PTAS is known~\cite{BenderMR04,ChekuriK02}.
The problem of minimizing (unweighted) $\ell_p$ norm of flow-times was studied by Im and Moseley~\cite{ImM17} who gave a constant factor approximation in polynomial time.

In the speed augmentation model introduced by Kalyanasundaram and Pruhs~\cite{KalyanasundaramP00},
the algorithm is given $(1+\eps)$-times extra speed than the optimal algorithm. Bansal and Pruhs~\cite{BansalP04} showed that Highest Density First (HDF) is $O(1)$-competitive for weighted $\ell_p$ norms of
flow-time for all values of $p \geq 1$.

The multi-cut problem on trees is known to be NP-hard, and a 2-approximation algorithm was given by Garg et al.~\cite{GargVY97}. As mentioned earlier, Chakrabarty et al.~\cite{ChakrabartyGK10} gave
a systematic study of column restricted covering integer programs (see also~\cite{BansalKS11} for follow-up results). The notion of shallow cell complexity for $0$-$1$ covering integer programs was formalized by Chan et al.~\cite{ChanGKS12}, where they relied on and generalized the techniques of Vardarajan~\cite{Varadarajan10}.
\section{Preliminaries}
An instance of the \wtdfl problem is specified by a set of $n$ jobs. Each job has a processing
requirement $p_j$, weight $w_j$ and release date $r_j$. We assume wlog that all of these quantities are integers,
and let $P$ denote the ratio of the largest to the smallest processing requirement of a job.
We divide the time line into unit length {\em slots} -- we shall often refer to the time slot $[t,t+1]$ as slot
$t$. A feasible schedule needs to process a job $j$ for $p_j$ units after its release
date. Note that we allow a job to be preempted. The weighted flow-time of a job is defined as $w_j \cdot (C_j-r_j)$,
  where $C_j$ is the slot in which the job $j$ finishes processing. The objective is to find a schedule which minimizes
the sum over all jobs of their weighted flow-time.

Note that any schedule would occupy exactly $T=\sum_j p_j$ slots. We say that a schedule is {\em busy} if it does not
leave any slot vacant even though there are jobs waiting to be finished. We can assume that the optimal schedule is
a busy schedule (otherwise, we can always shift some processing back and improve the objective function). We also assume that any busy schedule
fills the slots in $[0, T]$ (otherwise, we can break it into independent instances satisfying this property).

We shall also consider a generalization of the multi-cut problem on trees, which we call the \dmc problem. Here, edges have cost and size, and demands are specified by
ancestor-descendant paths. Each such path has a demand, and the goal is to select a minimum cost subset of edges
such that for each path, the total size of selected edges in the path is at least the demand of this path.

In Section~\ref{sec:ip}, we describe a well-known integer program for \wtdfl. This IP has variables $x_{j,t}$ for every job $j,$ and time $ t\geq r_j$, and it is supposed
to be 1 if $j$ completes processing after time $t$. The constraints in the IP consist of several covering constraints. However, there is an additional
complicating factor that $x_{j,t} \leq x_{j,t-1}$ must hold for all $t \geq r_j$. To get around this problem, we propose a different IP in Section~\ref{sec:newip}.
In this IP, we define variables of the form $y(j,S)$, where $S$ are exponentially increasing intervals starting from the release date of $j$.
This variable indicates whether $j$ is alive during the entire duration of $S$.
The idea is that if the flow-time of $j$ lies between $2^i$ and $2^{i+1}$, we can count $2^{i+1}$ for it, and say that $j$ is alive during the entire period
$[r_j + 2^i, r_j + 2^{i+1}]$. Conversely, if the variable $y(j,S)$ is 1 for an interval of the form $[r_j + 2^i, r_j + 2^{i+1}]$, we can assume (at a factor 2 loss) that
it is also alive during $[r_j, r_j+2^i]$. This allows us to decouple the $y(j,S)$ variables for different $S$.
By an additional trick, we can ensure that these intervals are laminar for different jobs. From here, the reduction to the \dmc problem is immediate (see Section~\ref{sec:red}
for details). In Section~\ref{sec:approx}, we show that the specific instances of \dmc obtained by such reductions have additional properties.
We use the property that the tree obtained from shortcutting all degree two vertices is binary and has $O(\log (nP))$ depth. We shall use the term {\em segment}
to define a maximal degree 2 (ancestor-descendant) path in the tree. So the property can be restated as -- any root to leaf path has at most $O(\log (nP))$ segments.
We give a dynamic programming algorithm for such instances. In the DP table for a vertex in the tree, we will look at a  sub-instance defined by the sub-tree below this vertex.
However, we also need to maintain
 the ``state'' of edges above it, where the state  means the ancestor edges selected by the algorithm. This would require too
much book-keeping. We use two ideas to reduce the size of this state -- (i) We first show that the optimum can be assumed to have certain smoothness properties, which
cuts down on the number of possible configurations. The smoothness property essentially says that the cost spent by the optimum on a segment does not vary
by more than a constant factor as we go to neighbouring segments,
(ii) If we could spend twice the amount spent by the algorithm on a segment $S$, and select low density edges, we could ignore the edges in a segment $S'$ lying above $S$ in the tree.

\subsection{An integer program}
\label{sec:ip}
We describe an integer program for the \wtdfl problem. This is well known (see e.g.~\cite{BansalP14}), but we give details for sake of completeness.
We will have binary variables $x_{j,t}$ for every job $j$ and time $t$, where $r_j \leq t \leq T$. This variable is meant to be 1 iff $j$ is {\em alive}
 at time $t$, i.e., its completion time is at least $t$. Clearly, the objective function is $\sum_j \sum_{t \in [r_j, T]} w_j x_{j,t} .$
 We now specify the constraints of the integer program. Consider a time interval $I=[s,t]$, where $0 \leq s \leq t \leq T$, and $s$ and $t$ are integers.
 Let $l(I)$ denote the length of this time interval, i.e., $t-s$. Let $J(I)$ denote the set of jobs released during $I$, i.e., $\{j: r_j \in I\}$,  and $p(J(I))$ denote the total processing time of jobs in $J(I)$. Clearly, the total volume occupied by jobs in $J(I)$ beyond $I$ must be at least $p(J(I))-l(I)$. Thus, we get the following integer program:~(IP1)
 \begin{align}
 \min & \sum_j \sum_{t \in [r_j, T]} w_j x_{j,t} \\
 \label{eq:knap}
 \sum_{j \in J(I)} x_{j,t} p_j & \geq  p(J(I)) - l(I) \ \ \ \ \mbox{ for all intervals $I=[s,t], 0 \leq s \leq t \leq T$} \\
 \label{eq:mon}
 x_{j,t} & \leq x_{j,t-1} \ \ \ \ \ \mbox{ for  all jobs $j$, and time $t$, $r_j < t \leq T$} \\
 \notag
x_{j,t} & \in \{0,1\} \ \ \ \mbox{for all $j,t$}
 \end{align}
 It is easy to see that this is a relaxation -- given any schedule, the corresponding $x_{j,t}$ variables will satisfy the constraints mentioned above, and
 the objective function captures the total weighted flow-time of this schedule. The converse is also true -- given any solution to the above integer program,
 there is a corresponding schedule of the same cost.

 \begin{theorem}
 \label{thm:lp}
 Suppose $x_{j,t}$ is a feasible solution to~(IP1). Then, there is a schedule for which the total weighted flow-time is equal to the cost of the solution
 $x_{j,t}$.
 \end{theorem}
 \begin{proof}
	We show how to build such a schedule. The integral solution $x$ gives us deadlines for each job. For a job $j$, define $d_j$ as one plus the last time $t$ such that $x_{j,t} = 1$. Note that $x_{j,t}=1$ for every $t \in [r_j, d_j)$. We would like to
	find a schedule which completes each job by time $d_j$~: if such a schedule exists, then the weighted flow-time of a job $j$ will be at most $\sum_{t \geq r_j} w_{j} x_{j,t}$, which is what we want.
	
	We begin by observing a simple property of a feasible solution to the integer program.
	
	\begin{claim}
		\label{cl:prop}
			Consider an interval $I=[s,t]$, $0 \leq s \leq t \leq T$. Let $J'$ be a subset of $J(I)$ such that $p(J') > l(I)$. If $x$ is a feasible solution to (IP1), then
			there must exist a job $j \in J'$ such that $x_{j,t} = 1$.
	\end{claim}
	
	\begin{proof}
		Suppose not. Then the LHS of constraint~(\ref{eq:knap}) for $I$ would be at most $p(J(I) \setminus J')$, whereas the RHS would be
		$p(J') + p(J(I) \setminus J') - l(I) > p(J(I) \setminus J')$, a contradiction.
	\end{proof}
It is natural to use the Earliest Deadline First rule to find the required schedule. We build the schedule from time $t=0$ onwards. At any time $t$, we say that a job $j$ is {\em alive}
	if $r_j \leq t$, and $j$ has not been completely processed by time $t$. Starting from time $t=0$, we process the alive job with earliest deadline $d_j$
	during $[t,t+1]$. We need to show that every job will complete before its deadline. Suppose not. Let $j$ be the job with the earliest deadline
	which is not able to finish by $d_j$.
	Let $t$ be first time before $d_j$ such that the algorithm processes a job whose deadline is more than $d_j$ during $[t-1,t]$, or
	it is idle during this time slot (if there is no such time slot, it must have busy from time $0$ onwards, and so set $t$ to 0).
	The algorithm processes jobs whose deadline is at most $d_j$ during $[t, d_j]$  -- call these jobs $J'$. We
	 claim that jobs in $J'$ were released after $t$ -- indeed if such a job was released before  time $t$, it would have been alive at time $t-1$ (since
	it gets processed after time $t$). Further its deadline is at most $d_j$, and so, the algorithm should not be processing a job whose deadline is more than
	$d_j$ during $[t-1,t]$ (or being idle). But now, consider the interval $I=[t, d_j]$.  Observe that $l(I) < p(J')$ -- indeed, $j \in J'$ and it is not completely processed
	during $I$, but the algorithm processes jobs from $J'$ only during $I$. Claim~\ref{cl:prop} now implies that there must be a job $j'$ in $J'$ for which
	$x_{j',d_j} = 1$ -- but then the deadline of $j'$ is more than $d_j$, a contradiction.
\end{proof}

\section{A Different Integer Program}
\label{sec:newip}

We now write a weaker integer program, but it has more structure in it. We first assume that $T$ is a power of 2 -- if not, we can pad the instance with a job of zero weight (this will
increase the ratio $P$ by at most a  factor $n$ only). Let $T$ be $2^\ell$.
We now divide the time line into nested dyadic segments. A dyadic segment is an interval of the form $[i \cdot 2^s, (i+1) \cdot 2^s]$ for some non-negative integers $i$ and $s$
(we shall use segments to denote such intervals to avoid any confusion with intervals used in the integer program).
 For $s=0, \ldots, \ell$,
we define $\S_s$ as the set of dyadic segments of length $2^s$  starting from 0, i.e., $\{[0, 2^s], [2^s, 2 \cdot 2^s], \ldots, [i \cdot 2^s, (i+1) \cdot 2^s],
\ldots, [T-2^s, T] \}$. Clearly, any segment of $\S_s$ is contained inside a unique segment of $\S_{s+1}$. Now, for every job $j$ we shall define a
sequence of dyadic segments $\Seg(j)$. The sequence of segments in $\Seg(j)$  partition the interval $[r_j,T]$. The construction
of $\Seg(j)$ is described in Figure~\ref{fig:seg} (also see the example in Figure~\ref{fig:segex}). It is easy to show by induction on $s$ that the parameter $t$ at the beginning of iteration $s$ in Step~2
of the algorithm is a multiple of $2^s$. Therefore, the segments added during the iteration for $s$ belong to $\S_s$. Although we do not specify for how long
we run the for loop in Step~2, we stop when $t$ reaches $T$ (this will always happen because $t$ takes values from the set of end-points in the segments
in $\cup_s \S_s$). Therefore the set of segments in $\Seg(j)$ are disjoint and cover $[r_j, T]$.

\begin{figure}[ht]
\begin{center}
\begin{boxedminipage}{5.8 in}
{\bf Algorithm FormSegments($j$)} \\
1. Initialize $t \leftarrow r_j$. \\
2. For $s= 0,1, 2, \ldots, $ \\
\hsp \hsp (i) If $t$ is a multiple of $2^{s+1}$, \\
\hsp \hsp \hsp \hsp  add the segments (from the set $\S_s$) $[t, t+2^s], [t+2^s, t+2^{s+1}]$ to $\Seg(j)$ \\
\hsp \hsp \hsp \hsp  update $t \leftarrow t+2^{s+1}$. \\
\hsp \hsp  (ii) Else add the segment (from the set $\S_s$) $[t,,t+2^s]$ to $\Seg(j)$. \\
\hsp \hsp \hsp \hsp update $t \leftarrow t+2^s$. \\
\end{boxedminipage}
\end{center}
\caption{Forming $\Seg(j)$. }
\label{fig:seg}
\end{figure}

\begin{figure}[ht]
\begin{center}
\input{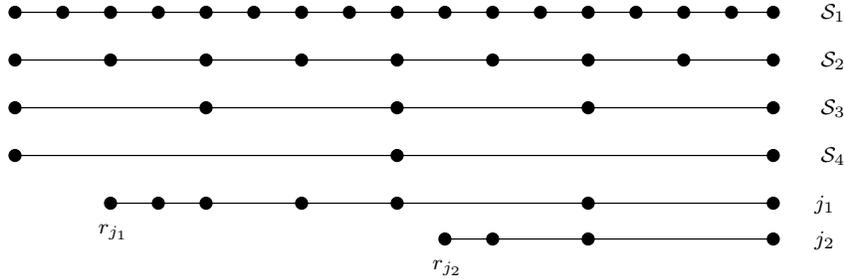}
\end{center}
\caption{The dyadic segments $\S_1, \ldots, \S_4$ and the corresponding $\Seg(j_1), \Seg(j_2)$ for two jobs $j_1, j_2$}
\label{fig:segex}
\end{figure}

For a job $j$ and segment $S \in \Seg(j)$, we shall  refer to the tuple $(j,S)$ as a {\em job-segment}. For a time $t$, we say that $t \in (j,S)$
(or $(j,S)$ contains $t$) if $[t,t+1] \subseteq S$.
We now show a crucial nesting property of these segments.
\begin{lemma}
\label{lem:partial}
Suppose $(j,S)$ and $(j',S')$ are two job-segments such that there is a time $t$ for which $t \in (j,S)$ and $t \in (j',S')$. Suppose $r_j \leq r_{j'}$, and
$S \in \S_s, S\ \in \S_{s'}$. Then $s \geq s'$.
\end{lemma}
\begin{proof}
We prove this by induction on $t$. When $t=r_{j'}$, this is trivially true because $s'$ would be 0. Suppose it is true for some $t \geq r_{j'}$.
Let $(j,S)$ and $(j',S') $ be the job segments containing $t$. Suppose $S \in \S_s, S' \in S_{s'}$. By induction hypothesis,
we know that $s \geq s'$. Let $(j', {\tS}')$ be the job-segment containing $t+1$, and let ${\tS}' \in \S_{\ts'}$ ($S'$ could be same as $\tS').$
We know that $\ts' \leq s'+1$. Therefore, the only interesting case is $s=s'$ and $\ts' = s'+1$. Since $s=s'$, the two segments $S$ and $S'$ must be
same (because all segments in $\S_s$ are mutually disjoint). Since $t \in S, t+1 \notin S$, it must be that $S=[l,t+1]$ for some $l$. The algorithm
for constructing $\Seg(j')$ adds  a segment from $\S_{s'+1}$ after adding $S'$ to $\Seg(j')$. Therefore $t+1$ must be a multiple of $2^{s'+1}$.
What does the algorithm for constructing $\Seg(j)$ do after adding $S$ to $\Seg(j)$? If it adds a segment from $\S_{s+1}$, then we are done again.
Suppose it adds a segment from $\S_s$. The right end-point of this segment would be $(t+1)+2^s$. After adding this segment, the algorithm would add a segment from $\S_{s+1}$ (as it cannot add more than 2 segments from $\S_s$ to $\Seg(j)$). But this can only happen if $(t+1) + 2^s$ is a multiple of
$2^{s+1}$ -- this is not true because $(t+1)$ is a multiple of $2^{s+1}$. Thus we get a contradiction, and so the next segment (after $S$) in $\Seg(j)$ must come from $\S_{s+1}$ as well.
\end{proof}

We now write a new IP.  The idea is that if a job $j$ is alive at some time $t$, then we will keep it alive during the entire duration of the segment in
$\Seg(j)$ containing $t$. Since the segments in $\Seg(j)$ have lengths in exponentially increasing order (except for two consecutive segments), this will
not increase the weighted flow-time by more than a constant factor.
For each job segment $(j,S)$ we have a binary variable $y(j,S)$, which is meant to be 1 iff the job $j$ is alive during the entire duration $S$.
For each job segment $(j,S)$, define its weight $w(j,S)$ as $w_j \cdot l(S)$ -- this is the contribution towards weighted flow-time of $j$ if $j$ remains
alive during the entire segment $S$.
We
get the following integer program~(IP2):
\begin{align}
 \min & \sum_j \sum_{s} w(j,S) y(j,S) \\
 \label{eq:cov}
 \sum_{(j,S): j \in J(I), t \in (j,S)} y(j,S) p_j & \geq  p(J(I)) - l(I) \ \ \ \ \mbox{ for all intervals $I=[s,t], 0 \leq s \leq t \leq T$} \\
 \notag
y(j,S) & \in \{0,1\} \ \ \ \mbox{for all job segments $(j,S)$}
 \end{align}

 Observe that for any interval $I$, the constraint~(\ref{eq:cov}) for $I$ has precisely one job segment for every job which gets released in $I$. Another interesting feature of this IP is that we do not have constraints corresponding to~(\ref{eq:mon}), and so it is possible that $y(j,S) = 1$ and $y(j,S')=0$ for two job segments $(j,S)$ and $(j,S')$ even
 though $S'$ appears before $S$ in $\Seg(j)$.
 We now relate the two integer programs.
 \begin{lemma}
 \label{lem:relate}
 Given a solution $x$ for~(IP1), we can construct a solution for~(IP2) of cost at most 8 times the cost of $x$. Similarly, given a solution $y$ for~(IP2), we can construct a solution for~(IP1) of cost at most 4 times the cost of $y$.
 \end{lemma}
 \begin{proof}
	Suppose we are given a solution $x$ for~(IP1). For every job $j$, let $d_j$ be the highest $t$ for which $x_{jt} =1$. Let the segments in $\Seg(j)$ (in the order they were added) be $S_1, S_2, \ldots$. Let $S_{i_j}$ be the segment in $\Seg(j)$ which contains $d_j$. Then we set $y(j,S_i)$ to 1
	for all $i \leq i_j$, and $y(j,S_i)$ to 0 for all $i > i_j$. This defines the solution $y$.  First we observe that $y$
	is feasible for~(IP2). Indeed, consider an interval $I=[s,t]$. If $x_{jt} = 1$ and $j \in J(I)$, then we do have
	$y(j,S) = 1$ for the job segment $(j,S)$ containing $t$. Therefore, the LHS of constraints~(\ref{eq:knap}) and~(\ref{eq:cov}) for $I$ are same. Also, observe that
	$$\sum_{S \in \Seg(j)} y(j,S) w(j,S) = \sum_{i=1}^{i_j} w_j \cdot l(S_i) \leq w_j 4 l(S_{i_j}), $$
	where the last inequality follows from the fact that there are at most two segments from any particular set $\S_s$ in
	$\Seg(j)$, and so, the length of every alternate segment in $\Seg(j)$ increases exponentially. So,
	$\sum_{i=1}^{i_j} l(s_i) \leq 2 \left( l(S_{i_j}) + l(S_{i_j-2}) + l(S_{i_j-4}) + \cdots \right) \leq 4 \cdot l(S_{i_j}).$
	Finally observe that $l(S_{i_j}) \leq 2 (d_j-r_j)$. Indeed, the length of $S_{i_{j-1}}$ is at least half of that of
	$S_{i_j}$. So, $$l(S_{i_j}) \leq 2 l(S_{i_j-1}) \leq 2 (d_j - r_j). $$
	Thus, the total contribution to the cost of $y$ from job segments corresponding to $j$ is at most $8w_j(d_j-r_j) =
	8 w_j \sum_{t \geq r_j} x_{j,t}. $
	This proves the first statement in the lemma.
	
	Now we prove the second statement.  Let $y$ be a solution to~(IP2). For each job $j$, let $S_{i_j}$ be the last job segment in $\Seg(j) = \{S_1, S_2, \ldots \}$ for which $y(j,S)$ is 1. We set $x_{j,t}$ to 1 for every $t \leq d_j$, where $d_j$ is the right end-point of
	$S_{i_j}$, and 0 for $t > d_j$. It is again easy to check that $x$ is a feasible solution to~(IP1). For a job $j$ the contribution of $j$ towards the cost of $x$ is
	$$w_j(d_j-r_j) = w_j \cdot \sum_{i=1}^{i_j} l(S_i) \leq 4 w_j \cdot l(S_{i_j}) \leq 4 \cdot \sum_{(j,S) \in \Seg(j)} w(j,S) y(j,S).$$
\end{proof}

 The above lemma states that it is sufficient to find a solution for~(IP2). Note that~(IP2) is a covering problem. It is also worth noting that the constraints~(\ref{eq:cov})
 need to be written only for those intervals $[s,t]$ for which  a job segment starts or ends at $s$ or $t$. 
 Since the number of job segments is $O(n \log T) = O(n \log (nP))$, it follows that~(IP2) can be turned into a polynomial size integer program.

 \section{Reduction to \dmc  on Trees}
 \label{sec:red}
 We now show that (IP2) can be viewed as a covering problem on trees. We define the covering problem, which we call Demand Multi-cut(\dmc) on trees. An instance $\I$ of this problem consists
 of a tuple $(\T, \P, c, p, d)$, where $\T$ is a rooted tree, and $\P$ consists of a set of ancestor-descendant paths. Each edge $e$ in $\T$ has a cost $c_e$ and size $p_e$. Each path $P in \P$ has
 a demand $d(P)$. Our goal is to pick a minimum cost subset of vertices $V'$ such that
for every path $P \in \P$, the set of vertices in $V' \cap P$ have total size at least $d(P)$.

We now reduce \wtdfl to \dmc on trees. Consider an instance $\I'$ of \wtdfl consisting of a set of jobs $J$. We reduce it to an instance $\I=(\T, \P, c,p,d)$ of \dmc. In our reduction,
$\T$ will be a forest instead of a tree, but we can then consider each tree as an independent problem instance of \dmc. 

We order the jobs in $J$ according to release dates (breaking ties arbitrarily) -- let $\prec_J$ be this total ordering (so, $j \prec_J j'$ implies that $r_j \leq r_{j'}$).
We now define the forest $\T$. The vertex set of $\T$ will
consist of all job segments $(j,S)$. For such a vertex $(j,S)$, let $j'$ be the job immediately preceding $j$ in the total order $\prec_J$. Since the job segments
in $\Seg(j')$ partition $[r_{j'}, T] $, and $r_{j'} \leq r_j$, there is a pair $(j', S')$ in $\Seg(j')$ such that
$S'$ intersects $S$, and so contains $S$, by Lemma~\ref{lem:partial}. We define $(j',S')$ as the parent of $(j,S)$. 
It is easy to see that this defines a forest structure, where the root vertices correspond to $(j,S)$, with $j$ being the first job in $\prec$. 
Indeed, if $(j_1, S_1), (j_2, S_2), \ldots, (j_k, S_k)$ is a sequence of nodes with $(j_i, S_i)$ being the parent of $(j_{i+1}, S_{i+1})$, then $j_1 \prec_J j_2 \prec_J \cdots \prec_J j_k$, 
and so no node in this sequence can be repeated. 

For each tree in this forest $\T$ with the root vertex being $(j,S)$, we add a new root vertex $r$ and make it the parent of $(j,S)$. 
We now define the cost and size of each edge. Let $e=(v_1, v_2)$ be an edge in the tree, where $v_1$ is the parent of $v_2$. Let $v_2$ correspond to the job segment $(j,S)$. 
Then $p_e = p_j$ and $c_e = w_e \cdot l(S)$. In other words, picking edge $e$ corresponds to selecting the job segment $(j,S)$. 

Now we define the set of paths $\P$. For each constraint~(\ref{eq:cov}) in~(IP2), we will add one path in $\P$. We first observe the following property. 
Fix an interval $I=[s,t]$ and consider the constraint~(\ref{eq:cov}) corresponding to it. Let $V_I$ be the vertices in $\T$ corresponding to the job segments appearing in the LHS of this constraint.
\begin{lemma}\label{lem:tree}
The vertices in $V_I$ form a path in $\T$ from an ancestor to a descendant.
\end{lemma}

\begin{proof}
	Let $j_1, \ldots, j_k$ be the jobs which are released in $I$ arranged according to $\prec_J$. Note that these
	will form a consecutive subsequence of the sequence obtained by arranging jobs according to $\prec_J$.
	Each of these jobs will have exactly one job segment $(j_i, S_i)$
	appearing on the LHS of this constraint (because for any such job $j_i$,  the segments in $\Seg(j_i)$ partition
	$[r_{j_i}, T]$). All these job segments contain $t$, and so, these segment intersect. Now, by construction of $\T$,
	it follows that the parent of $(j_i, S_i)$ in the tree $\T$ would be $(j_{i-1}, S_{i-1})$. This proves the claim.

\end{proof}

Let the vertices in $V_I$ be $v_1, \ldots, v_k$ arranged from ancestor to descendant. Let $v_0$ be the parent of $v_1$ (this is the reason why we added an extra root to each tree -- just in case
$v_1$ corresponds to the first job in $\prec_J$, it will still have a parent). We add a path $P_I = v_0, v_1, \ldots, v_k$ to $\P$ -- Lemma~\ref{lem:tree} guarantees that this will be an ancestor-descendant path. The demand $d(P)$ of this path is the quantity in the RHS of the corresponding constraint~(\ref{eq:cov}) for the interval $I$. The following claim is now easy to check. 

\begin{claim}
\label{cl:reduction}
Given a solution $E$ to the \dmc instance $\I$, there is a solution to~(IP2) for the instance $\I'$ of the same objective function value as that of $E$. 
\end{claim}
\begin{proof}
Consider a solution to $\I$ consisting of a set of edges $E$. For each edge $e=(v_1, v_2) \in E$ where $v_2=(j,S)$ is the child of $v_1$, we set $y(j,S)=1$. For rest of the job segments $(j,S)$, define
$y(j,S)$ to be 0. Since the cost of such an edge $e$ is equal to $w(j,S)$, it is easy to see that the two solutions have the same cost. Feasibility of~(IP2) also follows directly from the manner
in which the paths in $\P$ are defined. 
\end{proof}

This completes the reduction from \wtdfl to \dmc. This reduction is polynomial time because number of vertices in $\T$ is equal to the number of job segments, which is $O(n \log (nP))$. Each 
path in $\P$ goes between any two vertices in $\T$, and there is no need to have two paths between the same pair of vertices. Therefore the size of the instance $\I$ is polynomial in the size
of the instance $\I'$ of \wtdfl.

\section{Approximation Algorithm for the \dmc problem}
\label{sec:approx}
In this section we give a constant factor approximation algorithm for the special class of  \dmc problems which arise in the
reduction from \wtdfl. To understand the special structure of such instances, we begin with some definitions. Let $\I = (\T, \P, c, p, d)$
be an instance of \dmc.    The {\em density} $\rho_e$  of an edge $e$ is defined as the ratio $c_e/p_e$.  Let $\redc(\T)$ denote the tree obtained from $\T$ by short-cutting all non-root degree 2 vertices
(see Figure~\ref{fig:reduced} for an example). There is a clear correspondence between the vertices of $\redc(\T)$ and the non-root vertices in $\T$ which do not have degree 2.
In fact, we shall use $V(\redc(\T))$ to denote the latter set of vertices. The reduced height of $\T$ is defined as the height of $\redc(\T)$. In this section, we prove the following result. 
We say that a (rooted) tree is binary if every node has at most 2 children.

\begin{figure}[ht]
\begin{center}
\epsfig{file=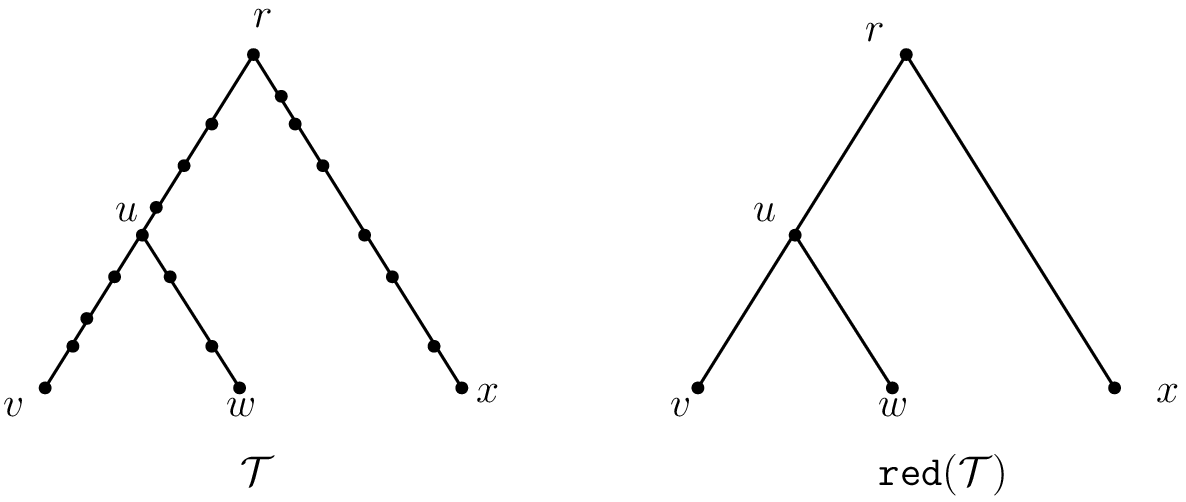, width=4in}
\end{center}
\caption{Tree $\T$ and the corresponding tree $\redc(\T)$. Note that the vertices in $\redc(\T)$ are also present in $\T$, and the
segments in $\T$ correspond to edges in $\redc(\T)$. The tree $\T$ has 4 segments, e.g., the path between $r$ and $u$. }
\label{fig:reduced}
\end{figure}

\begin{theorem}
\label{thm:dmc}
There is a constant factor approximation algorithm for instances $\I = (\T, \P, c, p, d)$ of \dmc where $\T$ is a binary tree. The running time of this algorithm is
$poly(n, 2^{O(H)}, \rho_{\max}/\rho_{\min})$, where $n$ denotes the number of nodes in $\T$, $H$ denotes the reduced height of $\T$, and $\rho_{\max}$ and $\rho_{\min}$ are the maximum and the minimum
density of an edge in $\T$ respectively.
\end{theorem}

\noindent
{\bf Remark:} In the instance $\I$ above, some edges may have 0 size. These edges are not considered while defining $\rho_{\max}$ and $\rho_{min}$.

Before we prove this theorem, let us see why it implies the main result in Theorem~\ref{thm:main}.

\noindent
\proofof{{\bf Theorem~\ref{thm:main}}}{
Consider  an instance $\I=(\T, \P, c, p, d)$ of $\dmc$ obtained via reduction from an instance $\I'$ of \wtdfl. Let $n'$ denote the number of jobs in $\I'$ and $P$ denote the ratio of the largest to the smallest job size in this instance. We had argued in the
previous section that $n$, the number of nodes in $\T$, is $O(n' \log P)$. We first perform some pre-processing on $\T$ such that the quantites $H, \rho_{\max}/\rho_{\min}$ do not become too large. 

\begin{itemize}
\item  Let $p_{\max}$ and $p_{\min}$ denote the maximum and the minimum size of a job  in the instance $\I'$. Each edge
    in $\T$ corresponds to a job interval in the instance $\I$. We select all edges for which the corresponding job interval has length at most $p_{\min}$.
    Note that after selecting these edges, we will contract them in $\T$ and adjust the demands of paths in $\P$ accordingly. 
    For a fixed job $j$, the total cost of such selected edges would be at most $4 w_j p_{\min} \leq 4 w_j p_j$ (as in the proof of Lemma~\ref{lem:relate},
    the corresponding job intervals have lengths which are powers of 2, and there are at most two intervals of the same length). Note that the cost of any optimal solution for $\I'$ is at least $\sum_j w_j p_j$, and
	so we are incurring an extra cost of at most 4 times the cost of the  optimal solution.

    So  we can assume that any edge in $\T$ corresponds to a job interval in $\I'$ whose length lies in the range $[p_{\min}, n' p_{\max}]$, because the length of the schedule is at most
    $n' p_{\max}$ (recall that we are assuming that there are no gaps in the schedule). 
\item Let $c_{\max}$ be the maximum cost of an edge selected by the optimal solution (we can cycle over
    all $n$ possibilities for $c_{\max}$, and select the best solution obtained over all such solutions). We remove (i.e., contract) all edges of cost more than $c_{\max}$, and select all edges
    of cost at most $c_{\max}/n$ (i.e., contract them and adjust demands of paths going through them) -- the cost of these selected edges will be at most a constant times the optimal cost. Therefore, we can assume that the costs of the edges lie in the range
    $[c_{\max}/n, c_{\max}]$. Therefore, the densities of the edges in $\T$ lie in the range $[\frac{c_{\max}}{n p_{\max}}, \frac{c_{\max}}{p_{\min}}]$.
\end{itemize}

Having performed the above steps, we now modify the tree $\T$ so that it becomes a binary tree. 
Recall that each vertex $v$ in $\T$ corresponds to a  dyadic interval $S_v$, and if $w$ is a child of $v$ then $S_w$ is contained in $S_v$ (for the root vertex, we 
can assign it the dyadic interval $[0,T]$). Now, consider a vertex $v$ with
	$S_v$ of size $2^s$ and suppose it has more than 2 children. Since the dyadic intervals for the children are mutually disjoint and contained in $S_v$, each of
	these will be of size at most $2^{s-1}$.
	Let $S_v^1$ and $S_v^2$ be the two dyadic intervals of length $2^{s-1}$ contained in $S_v$. Consider $S_v^1$. Let $w_1, \ldots, w_k$ be the children
	of $v$ for which the corresponding interval is contained in $S_v^1$. If $k > 1$, we create a new node $w$ below $v$
	(with corresponding interval being $S_v^1$) and make $w_1, \ldots, w_k$ children of $v$. The cost and size of the edge $(v,w)$ is 0. We proceed similarly
	for $S_v^2$. Thus, each node will now have at most 2 children. Note that we will blow up the number of vertices by a factor 2 only.

    We can now estimate the reduced height $H$ of $\T$. Consider a root to leaf path in $\redc(\T)$, and let the
    vertices in this path be $v_1, \ldots, v_k$.  Let $e_i$ denote the parent of $v_i$. Since each $v_i$ has two children in $\T$,
	the job interval corresponding to $e_i$ will be at least twice that for $e_{i+1}$.  From the first preprocessing step above, it follows that 
the length of this path is bounded by $\log (n'P)$, where $P$ denotes $p_{\max}/p_{min}$.
   Thus, $H$ is $O(\log (n'P))$.
 It  now follows from Theorem~\ref{thm:dmc} that we can get a constant factor approximation algorithm for the instance $\I$ in $poly(n,P)$ time.
}

We now prove Theorem~\ref{thm:dmc} in rest of the paper.
\subsection{Some Special Cases}
To motivate our algorithm, we consider some special cases first. Again, fix an instance $\I = (\T, \P, c,p,d)$ of \dmc. Recall that the tree $\redc(\T)$ is obtained by short-cutting all
degree 2 vertices in $\T$. Each edge in $\redc(\T)$ corresponds to a path in $\T$ -- in fact, there are maximal paths in $\T$ for which all internal nodes have degree 2. We call such paths
{\em segments} (to avoid confusion with paths in $\P$). See Figure~\ref{fig:reduced} for an example. Thus, there is a 1-1 correspondence between edges in $\redc(\T)$ and segments in $\T$. Recall that
every vertex in $\redc(\T)$ corresponds to a vertex in  $\T$ as well, and we will use the same notation for both the vertices.

\begin{figure}[ht]
\begin{center}
\epsfig{file=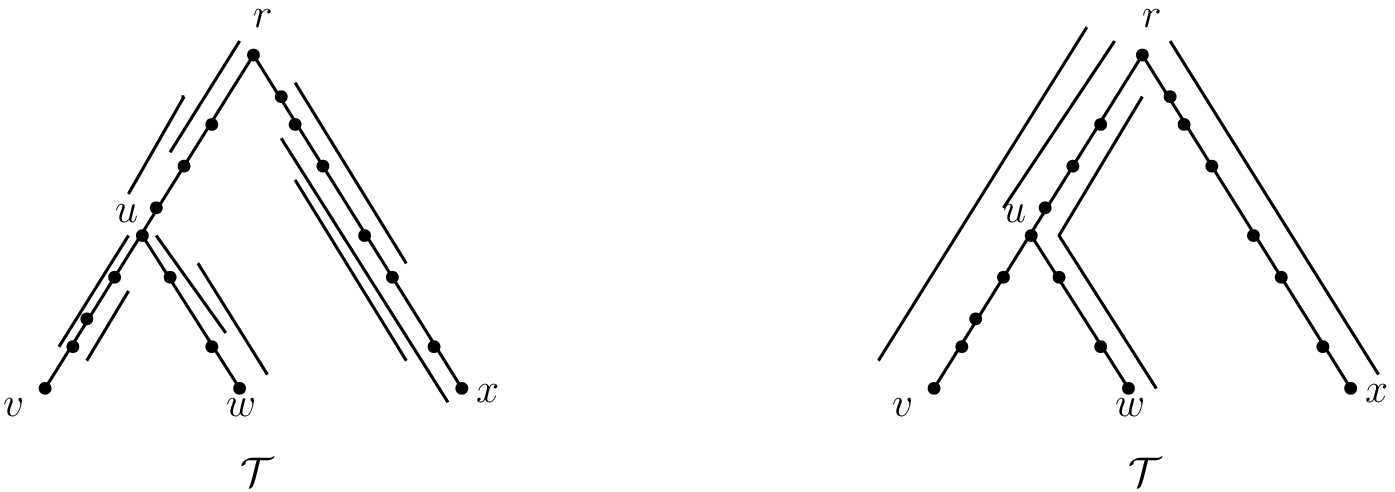, width=5in}
\end{center}
\caption{The left instance represents a segment confined instance whereas the right one is a segment spanning instance. }
\label{fig:special}
\end{figure}

\subsubsection{Segment Confined Instances}
\label{sec:confined}
The instance $\I$ is said to be {\em segment confined} if all paths in $\P$ are confined to one segment, i.e., for every path $P \in \P$, there is a segment $S$ in $\T$ such that the
edges of $P$ are contained in $S$. An example is shown in Figure~\ref{fig:special}. In  this section, we show that one can obtain constant factor polynomial time approximation algorithms for such instances.
In fact, this result follows from prior work on column restricted covering integer programs~\cite{ChakrabartyGK10}. Since each path in $\P$ is confined to one segment, we can think of this
instance as several independent instances, one for each segment. For a segment $S$, let $\I_S$ be the instance obtained from $\I$ by considering edges in $S$ only and
the subset $\P_S \subseteq \P$ of paths
which are contained in $S$. We show how to obtain a constant factor approximation algorithm for $\I_S$ for a fixed segment $S$.

Let the edges in $S$ (in top to down order) be $e_1, \ldots, e_m$.
The following integer program~(IP3) captures the \dmc problem for $\I_S$:
\begin{align}
\min & \sum_{e \in S } c_e x_e \\
\label{eq:cc}
\sum_{e \in P} p_e x_e & \geq  d(P) \ \ \ \ \mbox{ for all paths $P \in \P_S$} \\
x_e & \in \{0,1\} \ \ \ \mbox{for all $e \in S$}
\end{align}

Note that this is a covering integer program (IP) where the coefficient of $x_e$ in each constraint is either 0 or $p_e$.
Such an IP comes under the class of Column Restricted Covering IP as described in~\cite{ChakrabartyGK10}.  Chakrabarty et al.~\cite{ChakrabartyGK10} show that
one can
obtain a constant factor approximation algorithm for this problem provided one can prove that the integrality gaps of the corresponding LP relaxations
for the following two special class of problems are constant: (i) $0$-$1$ instances,  where the $p_e$ values are  either 0 or 1, (ii) priority
versions,
 where paths in $\P$ and edges have priorities (which can be thought  of as positive integers), and the selected edges
satisfy the property that for each path $P \in \P_S$, we  selected at least one edge  in it of priority at least that of $P$  (it is
easy to check that this is a special case of \dmc problem by assigning exponentially increasing demands to
paths of increasing priority, and similarly for edges).

Consider the class of $0$-$1$ instances first. We need to consider only those edges for which $p_e$ is 1 ( contract the
edges  for which $p_e$ is 0). Now observe that the constraint matrix on the LHS in~(IP3) has consecutive ones property
(order the paths in $\P_S$ in increasing order of left end-point and write the constraints in this order). Therefore,
the  LP relaxation has integrality gap of 1.

\noindent
{\bf Rounding the Priority Version}
We now consider the priority version of this problem. For each edge $e \in S$, we now have an associated priority $p_e$ (instead of size), and each path in $\P$ also has a priority demand $p(P)$, instead of its demand.   We need to argue about the integrality gap of the
following LP relaxation:

\begin{align}
\min & \sum_{e \in S} c_e x_e \\
\label{eq:pcc}
\sum_{e \in P: p_e \geq p(P)} x_e & \geq 1 \ \ \ \ \mbox{ for all paths $P \in \P_S$} \\
x_e & \geq 0 \ \ \ \mbox{for all $e \in S$}
\end{align}

We shall use the notion of {\em shallow cell complexity} used in~\cite{ChanGKS12}.
Let $A$ be the constraint matrix on the LHS above.
We first notice the following property of $A$.
\begin{claim}
	\label{cl:union}
	\global\long\def\clunion{
		Let $A^\star$ be a subset of $s$ columns of $A$. For a parameter $k, 0 \leq k \leq s$, there are at most $k^2 s$ distinct rows in $A^\star$
		with $k$ or fewer 1's (two rows of $A^\star$ are distinct iff they are not same as row vectors).
	}\clunion
\end{claim}

\begin{proof}
	Columns of $A$ correspond to edges in $S$. Contract all edges which are not in $A^\star$. Let $S^\star$ be the remaining (i.e., uncontracted) edges in $S$.
	Each path in $\P_S$ now maps to a new path obtained by contracting these edges. Let $\P^\star$ denote the set of resulting paths. For a path $P \in \P^\star$,
	let $E(P)$ be the edges in $P$ whose priority is at least that of $P$. In the constraint matrix $A^\star$, the constraint for the path $P$ has 1's in exactly the edges
	in $E(P)$.  We can assume that the set $E(P)$ is distinct for every path $P \in \P^\star$ (because
	we are interested in counting the number of paths with distinct sets $E(P)$).
	
	Let $\P^\star(k)$ be the paths in $\P^\star$ for which $|E(P)| \leq k$. We need to  count the cardinality of this set. Fix an edge $e \in S^\star$, let
	$S^\star(e)$ be the edges in $S^\star$ of priority at least that of $e$. Let $P$
	be a path in $\P^\star(k)$ which has $e$ as the least priority edge in $E(P)$ (breaking ties arbitrarily). Let $e_l$ and $e_r$ be the leftmost and the rightmost
	edges in $E(P)$ respectively. Note that $E(P)$ is exactly  the edges in $S^\star(e)$ which lie between $e_l$ and $e_r$. Since there are at most $k$ choices for
	$e_l$ and $e_r$ (look at the $k$ edges to the left and to the right of $e$ in the set $S^\star(e)$), it follows that there are at most $k^2$ paths $P$ in
	$\P^\star(k)$ which have $e$ as the least priority edge in $E(P)$. For every path in $\P^\star(k)$, there are at most $|E^\star| = s$ choices for the least priority edge.
	Therefore the size of $\P^\star(k)$  is at most $sk^2$.
\end{proof}

In the notation of~\cite{ChanGKS12}, the shallow cell complexity of this LP relaxation is $f(s,k) = sk^2$.
It now follows from Theorem~1.1 in~\cite{ChanGKS12} that the integrality gap of the LP relaxation for the priority version is a constant. Thus we obtain a constant factor approximation 
algorithm for segment restricted instances. 
\subsubsection{Segment Spanning Instances on Binary Trees}
\label{sec:segspan}
We now consider instances $\I$ for which each path $P \in \P$ starts and ends at the end-points of a segment, i.e., the starting or ending vertex of $P$ belongs to the set of vertices
in $\redc(\T)$. An example is shown in Figure~\ref{fig:special}. Although we will not use this result in the algorithm for the general case, many of the ideas will get extended to the general case.
We will use dynamic programming. For a vertex $v \in \redc(\T)$, let $\T_v$ be the sub-tree of $\T$ rooted below $v$ (and including $v$). Let $\P_v$ denote the subset of $\P$ consisting of those
paths which contain at least one edge in $\T_v$. By scaling the costs of edges, we will assume that the cost of the optimal solution lies in the range $[1,n]$ -- if
 $c_{\max}$ is the maximum cost of an edge selected by the optimal algorithm, then its cost lies in the range $[c_{\max}, n c_{\max}]$.

Before stating the dynamic programming algorithm, we give some intuition for the DP table. We will  consider sub-problems which correspond to covering paths in $\P_v$ by
edges in $\T_v$ for every vertex $v \in \redc(\T)$. However, to solve this sub-problem, we will also need to store the edges in $\T$ which are ancestors of $v$ and  are selected by our algorithm. Storing all such subsets would lead to too many DP table entries. Instead, we will work with the following idea -- for each segment $S$, let $B^\opt(S)$ be the total cost of edges in $S$ which
get selected by an optimal algorithm. If we know $B^\opt(S)$, then we can decide which edges in $S$ can be picked. Indeed, the optimal algorithm will solve a knapsack cover problem -- for the segment $S$,
it will pick edges of maximum total size subject to the constraint that their total cost is at most $B^\opt(S)$ (note that we are using the fact that every path in $\P$ which includes an edge in $S$ must
include all the edges in $S$). Although knapsack cover is NP-hard, here is a simple greedy algorithm which exceeds the budget $B^\opt(S)$ by a factor of 2, and does as well as the optimal solution (in terms
of total size of selected edges) -- order the edges in $S$ whose cost is at most $B^\opt(S)$ in order of increasing density. Keep selecting them in this order till we exceed the budget $B^\opt(S)$. Note that
we pay at most twice of $B^\opt(S)$ because the last edge will have cost at most $B^\opt(S)$. The fact that the total size of selected edges is at least that of the corresponding optimal value follows from
standard greedy arguments.

Therefore, if $S_1, \ldots, S_k$ denote the segments which lie above $v$ (in the order from the root to $v$), it will suffice if we store $B^\opt(S_1), \ldots, B^\opt(S_k)$ with the DP table entry for
$v$. We can further cut-down the search space by assuming that each of the quantities $B^\opt(S)$ is a power of 2 (we will lose only a multiplicative 2 in the cost of the solution). Thus, the total number of possibilities for $B^\opt(S_1), \ldots, B^\opt(S_k)$ is $O(\log^k n)$, because each of the quantities $B^\opt(S_i)$ lies in the range $[1,2n]$ (recall that we had assumed that the optimal value lies in the range
$[1,n]$ and now we are rounding this to power of 2). This is at most $2^{O(H \log \log n)}$, which is still not polynomial in $n$ and $2^{O(H)}$. We can further reduce this by assuming that for any two consecutive segments
$S_i, S_{i+1}$, the quantities $B^\opt(S_i)$ and $B^\opt(S_{i+1})$ differ by a factor of at most 8 -- it is not clear why we can make this assumption, but we will show later that this does leads to a constant
factor loss only. We now state the algorithm formally.

\noindent
{\bf Dynamic Programming Algorithm}

\newcommand{\grselect}{{\tt GreedySelect}}
We first describe the greedy algorithm outlined above. The algorithm {\grselect} is given in Figure~\ref{fig:greedy}.

\begin{figure}[ht]
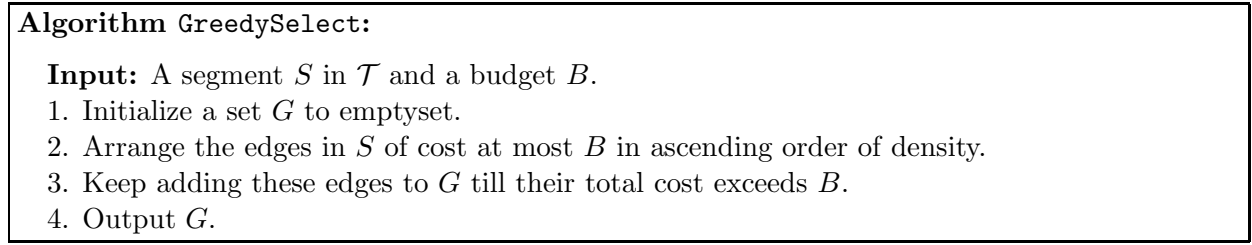

  \begin{center}
    \begin{boxedminipage}{6.5in}
      {\bf Algorithm \grselect  :} \medskip\\
         \sp \sp {\bf Input:} A segment $S$ in $\T$ and a budget $B$.   \\
         \sp \sp 1. Initialize a set $G$ to emptyset. \\
 		\sp \sp  2.  Arrange the edges in $S$ of cost at most $B$ in ascending order of density. \\
		\sp \sp 3. Keep adding these edges to $G$ till their total cost exceeds $B$. \\
		\sp \sp 4. Output $G$.
      \end{boxedminipage}
      \caption{Algorithm {\grselect} for selecting edges in a segment $S$ with a budget $B$. }
      \label{fig:greedy}
      \end{center}
\end{figure}

For a vertex $v \in \redc(\T)$, define the {\em reduced depth} of $v$ as its  at depth in $\redc(\T)$  (root has reduced depth 0). We say that a sequence $B_1, \ldots, B_k$  is a {\em valid state sequence} at a vertex $v$ in $\redc(\T)$ with reduced depth $k$ if it satisfies the following conditions:
\begin{itemize}
\item For all $i=1, \ldots, k$, $B_i$ is a power of 2 and  lies in the range $[1, 2n]$.
\item For any $i = 1, \ldots, k-1$, $B_i/B_{i+1}$ lies in the range $[1/8, 8]$.
\end{itemize}
If $S_1, \ldots, S_k$  is the sequence of segments visited while going from the root to $v$, then $B_i$ will correspond to $S_i$.

Consider a vertex $v \in \redc(\T)$ at reduced depth $k$, and a child $w$ of $v$ in $\redc(\T)$ (at reduced depth $k+1$). Let
$\Lambda_v = (B_1, \ldots, B_k)$ and $\Lambda_w = (B_1', \ldots, B_{k+1}')$ be valid state sequences at these two vertices
respectively. We say that $\Lambda_w$ is an {\em extension} of $\Lambda_v$ if $B_i = B_i'$ for $i=1, \ldots, k$. In the dynamic program, we maintain a table entry $T[v, \Gamma_v]$ for each vertex $v$ in $\redc(\T)$ and valid state sequence $\Gamma_v$ at $v$. Informally, this table entry stores the following quantity. Let $S_1, \ldots, S_k$ be the segments from the root to the vertex $v$. This table entry stores  the minimum cost of a subset $E'$ of edges in $\T_v$ such that
$E' \cup G(v)$ is a feasible solution for the paths in $\P_v$, where $G(v)$ is the union of the set of edges selected by {\grselect} in the segments
$S_1, \ldots, S_k$ with budgets $B_1, \ldots, B_k$ respectively.

The algorithm is described in Figure~\ref{fig:DPeasy}. We first compute the set $G(v)$ as outlined above. Let the children
of $v$ in the tree $\redc(\T)$ be $w_1$ and $w_2$. Let the segments corresponding to $(v, w_1)$ and $(v,w_2)$  be $S^1_{k+1}$ and
$S^2_{k+1}$ respectively. For both these children, we find the best extension of $\Gamma_v$.  For the node $w_r$,
we try out all possibilities for the budget $B_{k+1}^r$ for the segment $S_{k+1}^r$. For each of these choices, we
select a set of edges in $S_{k+1}^r$ as given by {\grselect} and
lookup the table entry for $w_r$ and the corresponding state sequence. We pick the choice for $B_{k+1}^r$ for which the combined
cost is smallest (see line~7(i)(c)).

 \begin{figure}[ht]
  \begin{center}
    \begin{boxedminipage}{6.5in}
      {\bf Fill DP Table  :} \medskip\\
         \sp \sp {\bf Input:} A node $v \in \redc(\T)$ at reduced  depth $k$, and a state sequence $\Lambda_v =  (B_1, \ldots, B_k).$ \\
        \sp \sp \sp \sp 0. If $v$ is a leaf node, set $D[v, \Lambda_v]$ to 0, and exit.  \\
         \sp \sp \sp \sp 1. Let $S_1, \ldots, S_k$ be the segments visited while going from the root to $v$ in $\T$. \\
          \sp \sp \sp \sp 2. Initialize $G(v) \leftarrow \emptyset$. \\
         \sp \sp \sp \sp 3. For $i=1, \ldots, k$ \\
         \sp \sp \sp \sp \sp \sp \sp(i)  Let $G_i(v)$ be the edges returned by \grselect($S_i, B_i$). \\
         \sp \sp \sp \sp \sp \sp \sp (ii) $G(v) \leftarrow G(v) \cup G_i(v). $ \\
        \sp \sp \sp \sp4. Let $w_1, w_2$ be the two children of $v$ in $\redc(\T)$ and \\
     \sp \sp \sp \sp  \sp \sp \sp \sp  the corresponding segments be $S_{k+1}^1, S^2_{k+1}$. \\
        \sp \sp \sp \sp 5. Initialize $M_1, M_2$ to $\infty$. \\
        \sp \sp \sp \sp 6. For $r = 1,2$ (go to each of the two children and solve the subproblems)\\
        \sp \sp \sp \sp \sp \sp \sp (i) For each extension $\Gamma_{w_r}  = (B_1, \ldots, B_k, B^r_{k+1})$ of $\Gamma_v$ do \\
        \sp \sp \sp \sp \sp \sp \sp \sp \sp \sp \sp (a) Let $G_{k+1}(w_r)$ be the edges returned by \grselect($S_{k+1}^r, B_{k+1}^r)$. \\
         \sp \sp \sp \sp \sp \sp \sp \sp \sp  \sp \sp (b) If any path in $\P_v$ ending in the segment $S_{k+1}^r$ is not satisfied  by $G(v) \cup G_{k+1}(w_r)$ \\
        \sp \sp \sp \sp \sp \sp \sp \sp \sp \sp \sp \sp \sp \sp \sp \sp \sp \sp  exit this loop \\
        \sp \sp \sp \sp \sp \sp \sp  \sp \sp  \sp \sp (c)   $M_r \leftarrow \min(M_r, {\mbox {cost of $G_{k+1}(w_r)$}}+  D[w_r, \Gamma_{w_r}]). $ \\
         \sp \sp \sp \sp 7. $D[v, \Lambda_v]  \leftarrow M_1 + M_2$.
      \end{boxedminipage}
      \caption{Filling a table entry $D[v, \Lambda_v]$  in the dynamic program. }
      \label{fig:DPeasy}
      \end{center}
\end{figure}

We will not analyze this algorithm here because  it's analysis will follow from the analysis of the more general case. We would like to remark that for any $v \in \redc(\T)$, the number of possibilities for a valid state sequence is bounded by $2^{O(H)} \cdot \log n$. Indeed,
there are $O(\log n)$ choices for $B_1$, and given $B_i$, there are only 7 choices for $B_{i+1}$ (since $B_{i+1}/B_i$ is a power of 2
and lies in the range $[1/8, 8]$). Therefore, the algorithm has running time polynomial in $n$ and $2^{O(H)}$.
\subsection{General Instances on Binary Trees}
We now consider general instances of \dmc on binary trees. We can assume that every path $P \in \P$ contains at least
one vertex of $\redc(\T)$ as an internal vertex. Indeed, we can separately consider the instance consisting of paths in $\P$ which
are contained in one of the segments -- this will be a segment confined instance as in Section~\ref{sec:confined}. We can get a constant
factor approximation for such an instance.

We will proceed as in the previous section, but now it is not sufficient to know the total cost spent by an optimal solution in each segment. For example, consider a segment $S$ which contains two edges $e_1$ and $e_2$; and  $e_1$ has low density, whereas $e_2$ has high density. Now, we would prefer to pick $e_1$, but it is possible that there are paths in $\P$ which contain $e_2$ but do not contain
$e_1$. Therefore, we cannot easily determine whether we should prefer picking $e_1$ over picking $e_2$. However, if all
edges in $S$ had the same density, then this would not be an issue. Indeed, given a budget $B$ for $S$, we would proceed as follows -- starting from each of the end-points, we will keep selecting edges of cost at most $B$ till their total cost exceeds $B$. The reason is that all edges are equivalent in terms of cost per unit size, and since each path in $\P$ contains at least one of the end-points of $S$, we might as well pick edges which are closer to the end-points. Of course, edges in $S$ may have varying density, and so, we will now need to know the budget spent by the optimum solution for each of the possible density values. We now describe this notion more formally.

\noindent
{\bf Algorithm Description}
We first assume that the density of any edge is a power of 128 -- we can do this by scaling the costs of edges by factors of at most 128.
We say that an edge $e$ is of {\em density class} $\tau$ if it's density is $128^\tau$. Let $\tau_{\max}$ and $\tau_{\min}$ denote the
maximum and the minimum density class of an edge respectively. Earlier, we had specified a budget $B(S)$ for each segment $S$ above $v$ while specifying the state at $v$. Now, we will need to store more information at every such segment. We shall use the term {\em cell} to refer to a pair $(S, \tau)$, where $S$ is a segment
and $\tau$ is a density class\footnote{For technical reasons, we will allow $\tau$ to lie in the range $[\tau_{\min}, \tau_{\max}+1]$}. Given a cell $(S, \tau)$ and a budget $B$, the algorithm~{\grselect} in Figure~\ref{fig:greedynew} describes the algorithm for selecting edges of density class $\tau$ from $S$. As mentioned above, this procedure ensures that we
pick enough edges from both end-points of $S$. The only subtlety is than in Step~4, we allow the cost to cross $2B$ -- the factor 2 is for technical reasons which will become clear later. Note that in Step~4 (and similarly in Step 5) we could end up selecting edges of total cost up tp $3B$ because each selected edge has cost at most $B$. 

\begin{figure}[ht]
  \begin{center}
    \begin{boxedminipage}{6.5in}
      {\bf Algorithm \grselect  :} \medskip\\
         \sp \sp {\bf Input:} A cell $(S, \tau)$ and a budget $B$.   \\
         \sp \sp 1. Initialize a set $G$ to emptyset. \\
 		\sp \sp  2.  Let $S(\tau)$ be the edges in $S$ of density class $\tau$ and cost at most $B$. \\
        \sp \sp 3. Arrange the edges in $S(\tau)$ from top to bottom order. \\
		\sp \sp 4. Keep adding these edges to $G$ in this order till their total cost exceeds $2B$. \\
\sp \sp 5. Repeat Step 4 with the edges in $S(\tau)$ arranged in bottom to top order.\\
		\sp \sp 6. Output $G$.
      \end{boxedminipage}
      \caption{Algorithm {\grselect} for selecting edges in a segment $S$ of density class $\tau$ with a budget $B$. }
      \label{fig:greedynew}
      \end{center}
\end{figure}

 As in the previous section, we define the notion of state
for a vertex $v \in \redc(\T)$.
Let $v$ be a node at reduced depth $k$  in $\redc(\T)$. Let $S_1, \ldots, S_k$ be the segments encountered as we go from the
root to $v$ in $\T$. If we were to proceed as in the previous section, we will store
a budget $B(S_i, \tau)$ For each cell $(S_i, \tau)$, $i=1, \ldots, k$,  $\tau \in [\tau_{\min}, \tau_{\max}]$. This will lead to a very large
number of possibilities (even if assume that for ``nearby'' cells, the budgets are not very different). Somewhat surprisingly, we show that it is enough to store this information at a small number of cells (in fact, linear in number of density classes and $H$).

\begin{figure}[ht]
\begin{center}
\epsfig{file=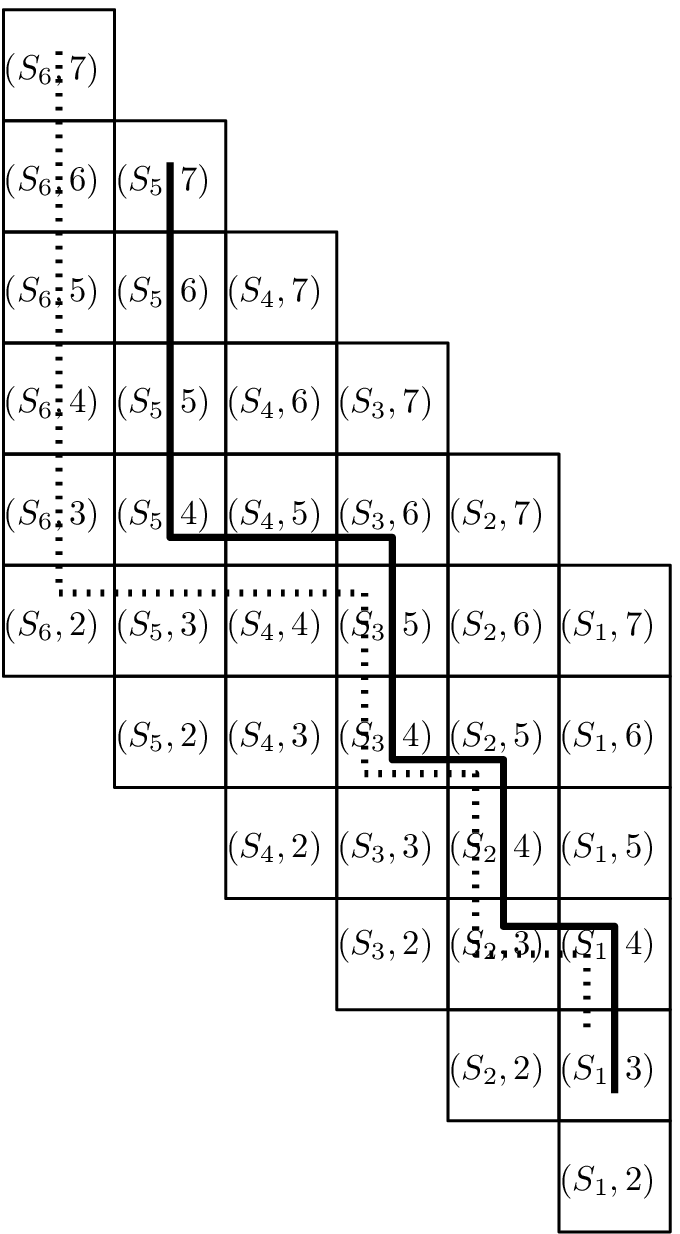, height=4 in}
\end{center}
\caption{$w$ is a vertex at reduced depth 6 and $v$ is the parent of $v$ in $\redc(\T)$. The segments above $w$ are labelled $S_1, \ldots, S_6$ (starting from the root downwards). The cells
are arranged in a tabular fashion as shown -- the density classes lie in the range $\{2, 3, \ldots, 7\}$. The solid line shows a valid cell sequence for $v$. The dotted line shows a valid cell sequence for
$w$ which is also an extension of the cell sequence for $v$ -- note that once the dotted line meets the solid line (in the cell $(S_3, 5)$, it stays with it till the end.  }
\label{fig:tab}
\end{figure}

To formalize this notion, we say that a sequence $\C_v= \sigma_1, \ldots, \sigma_{\ell}$ of cells is a {\em valid cell sequence} at $v$ if the following conditions are satisfied : (i) the first cell $\sigma_1$ is $(S_k, \tau_{\max})$, (ii) the last cell is of the form $(S_1, \tau)$
for some density class $\tau$, and (iii) if $\sigma=(S_i, \tau)$ is a cell in this sequence, then the next cell is either $(S_i, \tau-1)$
or $(S_{i-1}, \tau+1)$. To visualize this definition,  we arrange the cells $(S_i, \tau)$ in the form of a table shown in Figure~\ref{fig:tab}.  For each segment $S_i$, we draw a column in the table with one
  entry for each cell $(S_i, \tau)$, with $\tau$ increasing as we go up. Further as we go right, we shift these columns one step down. So row $\tau$ of this
  table will correspond to cells $(S_k, \tau), (S_{k-1}, \tau+1), (S_{k-2}, \tau+2)$ and so on. With this picture in mind, a a valid  sequence of cells starts from the top left and at each step it either goes one step down or one step right. Note that for such a sequence $\C_v$ and a segment $S_i$, the cells $(S_i, \tau)$ which appear in $\C_v$ are given by $(S_i, \tau_1), (S_i, \tau_1 + 1),
  \ldots, (S_i, \tau_2)$ for some $\tau_1 \leq \tau_2$. We say that the cells $(S_i, \tau)$, $\tau < \tau_1$, lie below the sequence $\C_v$, and the cells $(S_i, \tau), \tau > \tau_2$ lie above this
  sequence (e.g., in  Figure~\ref{fig:tab}, the cell $(S_2, 6)$ lies above the shown cell sequence, and $(S_4, 2)$ lies below it).

Besides a valid cell sequence, we need to define two more sequences for the vertex $v$:
\begin{itemize}
\item Valid Segment Budget Sequence: This is similar to the sequence defined in Section~\ref{sec:segspan}. This is a sequence $\Lambda^\seg_v := (B^\seg_1, \ldots, B^\seg_k)$, where $B^\seg_i$ corresponds to the segment $S_i$. As before, each of these quantities is a power of 2 and lies in the range $[1,2n]$. Further, for any $i$, the ratio $B^\seg_i/B^\seg_{i+1}$ lies in the range $[1/8,8]$.
\item Valid Cell Budget Sequence: Corresponding to the valid cell sequence $\C_v= \sigma_1, \ldots, \sigma_{\ell}$ and valid budget sequence $(B^\seg_1, \ldots, B^\seg_k)$, we have a sequence $\Lambda^\cell_v := (B^\cell_1, \ldots, B^\cell_{\ell})$,
where $B^\cell_j$ corresponds to the cell $\sigma_j$. Each of the quantities $B^\cell_j$ lies in the range $[1,2n]$. Further.
the ratio $B^\cell_j/B^\cell_{j+1}$ lies in the range $[1/8,8]$.
\end{itemize}

Intuitively, $B^\seg_i$ is supposed to capture the cost of edges picked by the optimal solution in $S_i$, whereas $B^\cell_j$, where $\sigma_j = (S_i, \tau)$,  captures the cost of the
density class $\tau$ edges in $S_i$ which get selected by the optimal solution.
A {\em valid state} $\State(v)$ at the vertex $v$ is given by the triplet $(\C_v, \Lambda^\seg_v, \Lambda^\cell_v)$ which in addition satisfies the following properties:
\begin{itemize}
\item[(i)] For a cell $\sigma_j = (S_i, \tau)$ in $\C_v$, the quantity $B^\cell_j \leq B_i^\seg$. Again, the intuition is clear -- the first quantity corresponds to cost of density class
$\tau$ edges in $S_i$, whereas the latter denotes the cost of all the edges in $S_i$ (which are selected by the optimal solution).\footnote{During the analysis, $B^\seg_i$ will be the {\em maximum} over all density classes $\tau$ of 
the  density class $\tau$ edges selected by the optimal solution from this segment. But this inequality will still hold.}
\end{itemize}

Informally, the idea behind these definitions is the following -- for each cell $\sigma_j$ in $\C_v$, we are given the corresponding budget $B^\cell_j$. We use this budget and Algorithm~{\grselect}
to select edges corresponding to this cell. For cells $\sigma = (S_i, \tau)$ which lie above $\C_v$, we do not have to use any edge of density class $\tau$ from $S_i$. Note that this does not mean  that our
algorithm will not pick any such edge, it is just that for the sub-problem defined by the paths in $\P_v$ and the state at $v$, we will not use any such edge (for covering a path in $\P_v$). For cells $\sigma = (S_i, \tau)$ which lie below
$\C_v$, we pick all edges of density class $\tau$ and cost at most $B^\seg_i$ from $S_i$. Thus, we can specify the subset of selected edges from $S_1, \ldots, S_k$ (for the purpose of covering paths in $\P_v$) by specifying these sequences only. The non-trivial fact is to show that maintaining such a small state (i.e., the three valid sequences) suffices to capture all scenarios. The algorithm for picking the edges for a specific segment $S$ is shown in Figure~\ref{fig:segmentselect}. Note one subtlety -- for the density class $\tau_1$ (in Step~4), we use budget $B^\seg_i$ instead
of the corresponding cell budget. The reason for this will become clear during the proof of Claim~\ref{cl:high}. 

\newcommand{\segmentselect}{{\tt SelectSegment}}
\begin{figure}[ht]
  \begin{center}
    \begin{boxedminipage}{6.5in}
      {\bf Algorithm \segmentselect  :} \medskip\\
         \sp \sp {\bf Input:} A vertex $v \in \redc(\T)$,  $\State(v) := (\C_v, \Lambda^\Seg_v, \Lambda^\cell_v)$, a segment $S_i$ lying above  $v$.   \\
         \sp \sp 1. Initialize a set $G$ to emptyset. \\
 		\sp \sp  2. Let $(S_i, \tau_1), (S_i, \tau_1 +1), \ldots, (S_i, \tau_2)$ be the cells in $\C_v$ corresponding to the segment $S_i$. \\
        \sp \sp  3. For $\tau = \tau_1+1, \ldots, \tau_2$ do \\
        \sp \sp \sp \sp \sp \sp (i) Add to $G$ the edges returned by \grselect($(S_i, \tau), B^\cell_j)$, \\
        \sp \sp \sp \sp \sp \sp \sp \sp \sp \sp  where $j$ is the index of $(S_i, \tau)$ in $\C_v$. \\
        \sp \sp 4. Add to $G$ the edges returned by \grselect($(S_i, \tau_1), B^\seg_i)$. \\
        \sp \sp  5.   Add to $G$ all edges $e \in S_i$ of density class strictly less than $\tau$ and for which $c_e \leq B^\seg_i$. \\
        \sp \sp 6. Return $G$.
      \end{boxedminipage}
      \caption{Algorithm {\segmentselect} for selecting edges in a segment $S$ as dictated by the state at $v$. The notations $B^\seg$ and $B^\cell$ are as explained in the text. }
      \label{fig:segmentselect}
      \end{center}
\end{figure}

Before specifying the DP table, we need to show what it means for a state to be an {\em extension} of another state. Let $w$ be a child of $v$ in $\redc(\T)$, and let $S_{k+1}$ be the corresponding segment joining $v$ and $w$. Given states $\State(v) := (\C_v, \Lambda^\Seg_v, \Lambda^\cell_v)$ and $\State(w) := (\C_w, \Lambda^\Seg_w, \Lambda^\cell_w)$, we say that $\State(w)$ is an extension
of $\State(v)$ if the following conditions are satisfied:
\begin{itemize}
\item If $\Lambda^\seg_v = (B^\seg_1, \ldots, B^\seg_k)$ and $\Lambda^\seg_w = (B^{\seg'}_1, \ldots, B^{\seg'}_{k+1}),$ then $B^\seg_i = B^{\seg'}_i$ for $i=1, \ldots, k$. In other words, the two
sequences agree on segments $S_1, \ldots, S_k$.
\item Recall that the first cell of $\C_w$ is $(S_{k+1}, \tau_{\max})$. Let $\tau_1$ be the smallest $\tau$ such that the cell $(S_{k+1}, \tau_1)$ appears in $\C_w$. Then the  cells succeeding $(S_{k+1},
\tau_1)$ in $\C_w$ must be of the form $(S_{k}, \tau_1+1), (S_{k-1}, \tau_1 + 2), \ldots,$ till we reach a cell which belongs to $\C_v$ (or we reach a cell for the segment $S_1$). After this the remaining cells in $\C_w$ are
the ones appearing in $\C_v$. Pictorially (see Figure~\ref{fig:tab}), the sequence for $\C_w$ starts from the top left, keeps going down till $(S_{k+1}, \tau_1)$, and then keeps moving right
till it hits $\C_v$. After this, it merges with $\C_v$.
\item The sequences $\Lambda^\cell_v$ and $\Lambda^\cell_w$ agree on cells which belong to both $\C_v$ and $\C_w$ (note that  the cells common to both will be a suffix of both the sequences). 
\end{itemize}

Having defined the notion of extension, the algorithm for filling the DP table for $D[v, \State(v)]$ is identical to the one in Figure~\ref{fig:DPeasy}. The details are given in Figure~\ref{fig:DP}. This completes the description of the algorithm.

\begin{figure}[ht]
  \begin{center}
    \begin{boxedminipage}{6.5in}
      {\bf Fill DP Table  :} \medskip\\
         \sp \sp {\bf Input:} A node $v \in \redc(\T)$ at reduced  depth $k$, $\State(v) = (\C_v, \Lambda^\Seg_v, \Lambda^\cell_v).$ \\
        \sp \sp \sp \sp 0. If $v$ is a leaf node, set $D[v, \State(v)]$ to 0, and exit.  \\
         \sp \sp \sp \sp 1. Let $S_1, \ldots, S_k$ be the segments visited while going from the root to $v$ in $\T$. \\
          \sp \sp \sp \sp 2. Initialize $G(v) \leftarrow \emptyset$. \\
         \sp \sp \sp \sp 3. For $i=1, \ldots, k$ \\
         \sp \sp \sp \sp \sp \sp \sp(i)  Let $G_i(v)$ be the edges returned by Algorithm~\segmentselect($v,S_i, \State(v)$). \\
         \sp \sp \sp \sp \sp \sp \sp (ii) $G(v) \leftarrow G(v) \cup G_i(v). $ \\
        \sp \sp \sp \sp4. Let $w_1, w_2$ be the two children of $v$ in $\redc(\T)$ and \\
     \sp \sp \sp \sp  \sp \sp \sp \sp  the corresponding segments be $S_{k+1}^1, S^2_{k+1}$. \\
        \sp \sp \sp \sp 5. Initialize $M_1, M_2$ to $\infty$. \\
        \sp \sp \sp \sp 6. For $r = 1,2$ (go to each of the two children and solve the subproblems)\\
        \sp \sp \sp \sp \sp \sp \sp (i) For each extension $\State(w_r)$ of  $\State(v)$ do \\
        \sp \sp \sp \sp \sp \sp \sp \sp \sp \sp \sp (a) Let $G_{k+1}(w_r)$ be the edges returned by \segmentselect($w_r, S_{k+1}^r, \State(w_r))$. \\
         \sp \sp \sp \sp \sp \sp \sp \sp \sp  \sp \sp (b) If any path in $\P_v$ ending in the segment $S_{k+1}^r$ is not satisfied  by $G(v) \cup G_{k+1}(w_r)$ \\
        \sp \sp \sp \sp \sp \sp \sp \sp \sp \sp \sp \sp \sp \sp \sp \sp \sp \sp  exit this loop \\
        \sp \sp \sp \sp \sp \sp \sp  \sp \sp  \sp \sp (c)   $M_r \leftarrow \min(M_r, {\mbox {cost of $G_{k+1}(w_r)$}}+  D[w_r, \State(w_r)]). $ \\
         \sp \sp \sp \sp 7. $D[v, \Lambda_v]  \leftarrow M_1 + M_2$.
      \end{boxedminipage}
      \caption{Filling a table entry $D[v, \State(v)]$  in the dynamic program. }
      \label{fig:DP}
      \end{center}
\end{figure}

\subsection{Algorithm Analysis}
We now analyze the algorithm.

\noindent
{\bf Running Time}

We bound the running time of the algorithm. First we bound the number of possible table entries.
\begin{lemma}
\label{cl:DPsize}
For any vertex $v$, the number of possible valid states is 
 $O\left((\log n)^2 \cdot 2^{O(H)} \cdot \left(\rho_{max}/\rho_{min}\right)^2 \right)$.
\end{lemma}
\begin{proof}
The length of a valid cell sequence is bounded by $(\tau_{\max} - \tau_{min})  + 2H$. To see this, fix a vertex $v$ at
reduced depth $k$, with segments $S_1, \ldots, S_k$ from the root to the vertex $v$. Consider a valid cell sequence $\sigma_1, \ldots, \sigma_\ell$. For a cell $\sigma_j = (S_i, \tau)$, define a potential $\Phi_j = j+2i+\tau$. We claim
that the potential $\Phi_j = \Phi_{j+1}$ for all indices $j$ in this sequence. To see this, there are two options for
$\sigma_{j+1}=(S_{i'}, \tau')$ :
\begin{itemize}
\item $S_i = S_{i'}, \tau' = \tau - 1$: Here, $\Phi_{j+1} = j+1 + 2i + \tau-1 = j + 2i + \tau = \Phi_j$.
\item $S_{i'} = S_{i-1}, \tau' = \tau+1$: Here $\Phi_{j+1} = j+1 + 2(i-1) + \tau + 1 = \Phi_j. $
\end{itemize}
Therefore,
$$ \ell+ \tau_{\min} \leq \Phi_l = \Phi_0 \leq 2H + \tau_{\max}. $$
It follows that $\ell \leq 2H + \tau_{\max} - \tau_{\min}$.
 Given the cell $\sigma_j$, there are only two choices for $\sigma_{j+1}$. So, the number of possible valid cell sequences
is bounded by $2^{2H + \tau_{\max} - \tau_{\min}} \leq 2^{2H} \cdot \rho_{\max}/\rho_{\min}$.

Now we bound the number of valid segment budget sequences. Consider such a sequence $B^\seg_1, \ldots, B^\seg_k$. Since
$B^\seg_1 \in [1, 2n]$ and it is a power of 2, there are $O(\log n)$ choices for it. Given $B^\seg_i$, there are at most 7 choices
for $B_{i+1}^\seg$, because $B^\seg_{i+1}/B^\seg_i$ is a power of 2 and lies in the range $[1/8,8]$. Therefore, the number
of such sequences is at most $O(\log n) \cdot 7^k \leq O(7^H \log n)$. Similarly, the number of valid cell budget sequences
is at  most $O(\log n) \cdot 7^\ell$, where $\ell$ is the maximum length of a valid cell sequence. By the argument above, $\ell$ is
at most $2H + \tau_{\max}-\tau_{\min}$. Combining everything, we see that the number of possible states for $v$ is bounded by
a constant times
$$ 2^{2H} \cdot \rho_{\max}/\rho_{\min} \cdot 7^H \log n \cdot \log n \cdot 7^{2H} \cdot \rho_{\max}/\rho_{\min}, $$
which implies the desired result.
\end{proof}

We can now bound the running time easily.
\begin{lemma}
\label{lem:running}
The running time of the algorithm is polynomial in $n, 2^{H}, \rho_{\max}/\rho_{\min}$.
\end{lemma}

\begin{proof}
To fill the table entry for $D[v, \Gamma(v)]$, where $v$ has  children $w_1, w_2$, the Algorithm in Figure~\ref{fig:DP}  cycles
through the number of possible extensions of $\Gamma(v)$ for each of the two children. Since any valid extension of $\Gamma(v)$
for a child $w_r$ is also a valid state at $w_r$, the result follows from Lemma~\ref{cl:DPsize}.
\end{proof}

\noindent
{\bf Feasibility}

We now argue that the table entries in the DP correspond to valid solutions.
  Fix a vertex $v \in \redc(\T)$, and let $S_1, \ldots, S_k$ be the segments as we go from the  root to $v$.
Recall that $\T(v)$ denotes the sub-tree of $\T$ rooted below $v$ and $\P(v)$ denotes the paths in $\P$ which have an
internal vertex in $\T(v)$. For a segment $S$ and density class $\tau$, let $S(\tau)$ denote the edges of density class
$\tau$ in $S$.

\begin{lemma}
\label{lem:correct}
Consider the algorithm in Figure~\ref{fig:DP} for filling the DP entry $D[v, \St(v)]$, and let $G(v)$ be the set of vertices
obtained after Step~3 of the algorithm.  Assuming that this table entry is not $\infty$, 
there is a subset $Y(v)$ of edges in $\T(v)$ such that
the cost of $Y(v)$ is at equal to   $D[v,\St(v)]$
and  $Y(v) \cup G(v)$ is a feasible solution for the paths in $\P(v)$.
\end{lemma}

\begin{proof}
	We prove this by induction on the reduced height of $v$.
 If $v$ is a leaf, then $\P(v)$ is empty, and so the result follows trivially. Suppose it is true for all nodes in $\redc(\T)$ at reduced height at most $k-1$, and $v$ be
	at height $k$ in $\redc(\T)$.
We use the notation in Figure~\ref{fig:DP}. Consider a child $w_r$ of $v$, where $r$ is either $1$ or $2$.
	Let the value of $M_r$ used in Step~7 be
	equal to the cost of $G_{k+1}(w_r) + D[w_r, \St(w_r)]$ for some $\St(w_r)$ given by $(\C_{w_r}, \Lambda^\seg_{w_r}, \Lambda^\cell_{w_r})$, with $\Lambda^\seg_v = (B^\seg_1, \ldots, B^\seg_k)$, $\Lambda^\seg_{w_r} = (B^{\seg'}_1, \ldots, B^{\seg'}_{k+1})$ and
	$\Lambda^\cell_{w_r} = (B^{\cell'}_1,\ldots,B^{\cell'}_{\ell})$.
	Let $G(v)$ and $G_{k+1}(w_r)$ be as in the steps~3 and~6(i)(a) respectively. We ensure that $G(v) \cup G_{k+1}(w_r)$ covers
	all paths in $\P(v)$ which end before $w_r$. The following claim is the key to the correctness of the algorithm.
	\begin{claim}
		Let $G(w_r)$  be edges obtained at the end of Step~3 in the algorithm in Figure~\ref{fig:DP} when filling the DP table entry $D[w_r, \St(w_r)]$.
Then $G(w_r)$ is a subset of $G(v) \cup G_{k+1}(w_r)$.
	\end{claim}
	\begin{proof}
Let $S_{k+1}^r$ be the segment between $v$ and $w_r$. By definition, $G_{k+1}(w_r)$ and $G(w_r) \cap S_{k+1}^r$ are
identical. Let us now worry about segments $S_i, i \leq k$.
Fix such a segment $S_i$.
		
		We know that after the cells corresponding to the segment $S^r_{k+1}$, 
the sequence $\C_{w_r}$ lies below $\C_v$ till it meets $\C_v$.
Now consider various case for an arbitrary cell $\sigma=(S_i, \tau)$ (we refer to the algorithm~\segmentselect in
Figure~\ref{fig:segmentselect}):
\begin{itemize}
\item The cell $\sigma$ lies above $\C_{w_r}$: $G(w_r)$ does not contain any edge of $S_i(\tau)$.
\item The cell $\sigma$ lies below $\C_{w_r}$: The cell will lie below $\C_{v}$ as well, and so, $G(v)$ and $G(w)$ will
contain the same edges from $S_i(\tau)$ (because $B^\seg_i = B^{\seg'}_i$ are same).
\item The cell $\sigma$ lies on $\C_{w_r}$: If it also lies on $\C_v$, then the fact that $B^\cell$ and $B^{\cell'}$ values for this
cell are same  implies that $G(v)$ and $G(w_r)$ pick the same edges from $S_i(\tau)$ (in case $\tau$ happens to be the smallest indexed density class for which $(S_i, \tau) \in \C_{w_r}$, then the same will hold for $\C_v$ as well). If it lies below $\C_v$, then the facts
that $B^\seg_i = B^{\seg'}_i$, and $B^\seg_i \geq B^\cell_j = B^{\cell'}_{j'}$, where $j$ and $j'$ are the indices of this cell in the two cell sequences respectively, imply that $G(v)$ will pick all the edges fof cost at most $B^\seg_{i}$ from $S_i(\tau)$, whereas $G(w_r)$
will pick only a subset of these edges.
\end{itemize}

We see that $G(v) \cap S_i$ contains $G(w_r) \cap S_i$. This proves the claim.
		
	\end{proof}
	
	By induction hypothesis, there is a subset $Y(w_r)$ of edges in the subtree $\T(w_r)$ of cost equal to  $D[w_r,\St(w_r)]$ such that $Y(w_r) \cup G(w_r)$ satisfies all paths in $\P(w_r)$.
	We already know that $G(v) \cup G_{k+1}(w_r)$ covers all paths in $\P(v)$ which end in the segment $S_{k+1}^r$.
	Since any path in $\P(v)$ will either end in $S^1_{k+1}$ or $S^2_{k+1}$, or
	will belong to $\P(w_1) \cup \P(w_2)$, it follows that all paths in $\P(v)$ are covered by $\cup_{r=1}^2 (Y(w_r) \cup G(w_r) \cup G_{k+1}(w_r)) \cup G(v)$. Now,
	the claim above shows that $G(w_r) \subseteq G_{k+1}(w_r) \cup G(v)$. So this set is same as $Y(w_1) \cup Y(w_2) \cup G_{k+1}(w_1) \cup G_{k+1}(w_2)\cup G(v)$ (and these sets are mutually disjoint).
Recall that $M_r$ is equal to the cost of $G_{k+1}(w_r) \cup Y(w_r)$, it follows that the
	the DP table entry for $v$ for these parameters is exactly the cost of $Y(w_1) \cup Y(w_2) \cup G_{k+1}(w_1) \cup G_{k+1}(w_2)$. This proves the lemma. 
\end{proof}

For the root vertex $r$, a valid state at $r$ must be the empty set. The above lemma specialized to the root $r$ implies:
\begin{corollary}
\label{cor:feas}
Assuming that $D[r, \emptyset]$ is not $\infty$, 
it is the cost of a feasible solution to the input instance.
\end{corollary}

\newcommand{\Opt}{{\tt OPT}}
\newcommand{\sB}{{B^\star}}
\noindent
{\bf Approximation Ratio}
Now we related the values of the DP table entries to the values of the optimal solution for suitable sub-problems.
We give some notation first. Let $\Opt$ denote an optimal solution to the input instance. For a segment $S$, we shall use
$S^\opt$ to the denote the subset of $S$ selected by $\Opt$. 
 Similarly, let $S^\opt(\tau)$ denote the subset of $S(\tau)$ selected by $\Opt$. 
 Let $B^\opt(S, \tau)$ denote the total cost of edges in $S^\opt(\tau)$, and $B^\opt(S)$ denote the
cost of edges in $S^\opt$.

We first show that we can upper bound $B^\opt(S, \tau)$ by values $\sB(S, \tau)$ values such that the latter values are close to each other for nearby cells. For two segments $S$ and $S'$, define the distance between them as the distance between
the corresponding edges in $\redc(T)$ (the distance between two adjacent edges is 1). Similarly, we say that a segment is the parent of another segment if this relation holds for the corresponding edges in $\redc(\T)$. Let $\cells(\T)$ denote the set
of all cells in $\T$.

\begin{lemma}
\label{lem:smooth}
We can find values $\sB(S, \tau)$ for each cell $(S, \tau)$ such that the following properties are satisfied: (i) for every cell $(S, \tau)$, $\sB(S, \tau)$
is a power of 2, and $\sB(S, \tau) \geq B^\opt(S, \tau)$, (ii)$ \sum_{(S, \tau) \in \cells(\T)} \sB(S, \tau) \leq 16 \cdot \sum_{(S, \tau) \in \cells(\T)}
B^\opt(S, \tau)$ and (iii) (smoothness) for every pair of segments $S, S'$, where $S'$ is the parent of $S$, and density class $\tau$,
$$8 \sB(S, \tau+1) \geq \sB(S, \tau) \geq \sB(S, \tau+1)/8, \mbox{ and } 8 \sB(S', \tau) \geq \sB(S, \tau) \geq \sB(S', \tau)/8. $$
\end{lemma}
\begin{proof}
We define $$\sB(S, \tau):= \sum_{i \geq 0} \sum_{S' \in N_i(S)} \sum_{j} \frac{B^\opt(S', \tau+i+j)}{4^{i+|j|}}, $$
where $i$ varies over non-negative integers, $j$ varies over integers and the range of $i,j$ are such  that $\tau+i+j$ remains a valid density class; and $N_i(S)$ denotes the segments which are at distance at most $i$ from $S$.
Note that $\sB(S, \tau)$ is a not a power of 2 yet, but we will round it up later. As of now, $\sB(S, \tau) \geq B^\opt(
S, \tau)$ because the term on RHS for $i=0, j=0,$
is exactly $B^\opt(S,\tau)$.

We now verify the second property. 
 We  add $\sB(S, \tau)$ for all the cells $(S, \tau)$.
Let us count the total contribution towards terms containing $B^\opt(S', \tau')$
on the RHS. For every segment $S \in N_i(S')$, and density class $\tau' - i - j$, it will receive a contribution of $\frac{1}{4^{i+|j|}}$. Since $|N_i(s')| \leq 2^{i + 1}$ (this is where we are using the fact that $\T$ is binary), this is at most
$$ \sum_{i \geq 0} \sum_j \frac{2^{i+1}}{4^{i+|j|}} \leq \sum_{i \geq 0} \frac{2^{i+2}}{ 4^i} \leq  8. $$

Now consider the third condition. Consider the expressions for $\sB(S, \tau)$ and $\sB(S', \tau)$ where $S'$ is the parent of $S$. If a segment is at distance $i$ from $S$, its distance
from $S'$ is either $i$ or $ i\pm 1$. Therefore, the coefficients of $B^\opt(S'', \tau'')$ in the expressions for $\sB(S, \tau)$ and $\sB(S', \tau)$ will differ by a factor of at most 4.
The same observation holds for $\sB(S, \tau)$ and $\sB(S, \tau+1)$. It follows that
$$4 \sB(S, \tau+1) \geq \sB(S, \tau) \geq \sB(S, \tau+1)/4, \mbox{ and } 4\sB(S', \tau) \geq \sB(S, \tau) \geq \sB(S', \tau)/4. $$

Finally, we round all the $\sB(S, \tau)$ values up to the nearest power of 2. We will lose an extra factor of 2 in the statements~(ii) and~(iii) above.
\end{proof}

We will use the definition of $\sB(S, \tau)$ in the lemma above for rest of the discussion. For a segment $S$, define $\sB(S)$ as the maximum over all density classes $\tau$ of
$\sB(S, \tau)$. The following corollary follows immediately from the lemma above.
\begin{corollary}
\label{cor:close}
Let $S$ and $S'$ be two segments in $\T$ such that $S'$ is the parent of $S$. Then $\sB(S)$ and $\sB(S')$ lie within factor of 8 of each other.
\end{corollary}
\begin{proof}
Let $\sB(S)$ be equal to $\sB(S, \tau)$ for some density class $\tau$. Then, $$ \sB(S) = \sB(S, \tau) \stackrel{Lemma~\ref{lem:smooth}}{\leq} 8 \sB(S', \tau)\leq 8 \sB(S'). $$
The other part of the argument follows similarly.
\end{proof}

\newcommand{\Sst}{\St^\star}
\newcommand{\Cs}{\C^\star}
\newcommand{\segs}{\Lambda^{\star \seg}}
\newcommand{\Cells}{\Lambda^{\star \cell}}
\newcommand{\bsegs}{B^{\star \seg}}
\newcommand{\bcells}{B^{\star \cell}}
The plan now is to define a valid state $\Sst(v)= (\Cs_v, \segs_v, \Cells_v)$ for each of the vertices $v$ in $\redc(\T)$. We begin by defining a {\em critical density}
$\tau^\star(S)$ for each segment $S$. Recall that $S(\tau)$ denotes the edges of density class $\tau$ in $S$. Let $S(\leq \tau)$ denote the edges class of density class at most
$\tau$ in $S$. For a density class $\tau$ and a budget $B$, let $S(\leq \tau, \leq B)$ denote the edges in $S(\leq \tau)$ which have cost at most $B$. Define $\tau^\star$ as
the smallest density class $\tau$ such that the total cost of edges in $S(\leq \tau, \leq \sB(S))$ is at least $4\sB(S) + \sum_{\tau' \leq \tau} B^\opt(S, \tau')$ (if no such density class exists, set $\tau$ to
$\tau_{\max}$). Intuitively, we are trying to augment the optimal solution by
low density edges, and $\tau^\star(S)$ tells us the density class till which we can essentially take all the edges in $S$ (provided we do not pick any edge which is too expensive).

\begin{figure}[ht]
\begin{center}
\epsfig{file=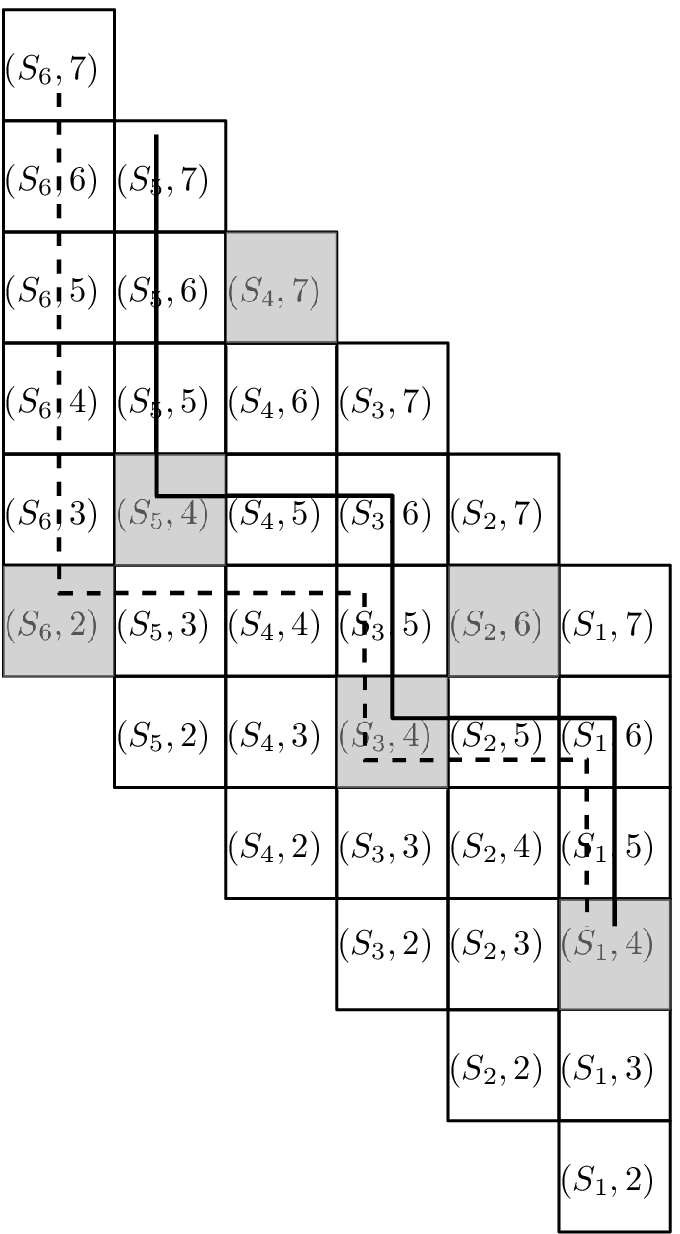, height=4 in}
\end{center}
\caption{Refer to the notation used in  Figure~\ref{fig:tab}. The shaded cells represnt the critical density class for the
corresponding segment. The solid line shows $\Cs_v$ and the dotted line shows $\Cs_w$. As an example, the
cell  $(S_4, 4)$ dominates the cells $(S_3, 6), (S_3, 7)$ and  $(S_2, 7)$.}
\label{fig:opt}
\end{figure}

Having defined the notion of critical density, we are now ready to define  a valid state $\Sst(v)$ for each vertex $v$ in $\redc(\T)$. Let $v$ be such a vertex at reduced
depth $k$ and let $S_1, \ldots, S_k$ be the segments starting from the root to $v$. Again, it is easier to see the definition of the cell sequence $\Cs_v$ pictorially. As in Figure~\ref{fig:opt}, the cell sequence starts with $(S_k,\tau_{\max})$ and keeps going down till it reaches the cell $(S_k, \tau^\star(S_k))$. Now it  keeps going right as long
as the cell corresponding to the critical density lies above it. If this cell lies below it, it moves down. The formal procedure for constructing this path is given in Figure~\ref{fig:path}. For sake of brevity, let $\tau_i^\star$ denote $\tau^\star(S_i)$.

\begin{figure}[ht]
  \begin{center}
    \begin{boxedminipage}{6.5in}
      {\bf Construct Sequence $\Cs_v $ :} \medskip\\
         \sp \sp {\bf Input:} A node $v \in \redc(\T)$ at depth $k$, integers  $\tau_1^\star, \ldots, \tau_k^\star$ \\
        \sp \sp \sp \sp 1. Initialise $\Cs_v$ to empty sequence, and $i \leftarrow k, \tau \leftarrow \tau_{\max}$ \\
         \sp \sp \sp \sp 2. While $(i \geq 1)$ \\
         \sp \sp \sp \sp \sp \sp \sp \sp (i) Add the cell $(S_i, \tau)$ to $\Cs_v$. \\
          \sp \sp \sp \sp \sp \sp \sp \sp (ii) If $\tau > \tau^\star_i$, $\tau \leftarrow \tau-1$ \\
          \sp \sp \sp \sp \sp \sp \sp \sp (iii) Else $i = i-1, \tau \leftarrow \tau+1.$ \\
      \end{boxedminipage}
      \caption{Construction of the path $\Cs_v$. }
      \label{fig:path}
      \end{center}
\end{figure}

We shall denote the sequence $\Cs_v$ by $\sigma^\star_1, \ldots, \sigma^\star_\ell$.
The corresponding segment budget sequence and cell budget sequences are easy to define. Define $\segs_v = (\bsegs_1, \ldots, \bsegs_k)$, where $\bsegs_i := B^\star(S_i)$.
Similarly, define $\Cells_v = (\bcells_1, \ldots, \bcells_{\ell})$ such that for the cell $\sigma^\star_j = (S_i, \tau)$, $\bcells_j := B^\star(S_i, \tau)$. This completes the
definition of $\Sst(v)$. It is easy to check that these are valid sequences. Indeed, Lemma~\ref{lem:smooth} and Corollary~\ref{cor:close} show that each of the quantities $\bsegs_i, \bcells_j$ are at most
$2 n$, and two such consecutive quantities are within factor of 8 of each other.

Further, let $w$ be a child of $v$ in $\redc(\T)$. It is again it is to see that $\Sst(w)$ is an extension of $\Sst(v)$. The procedure for constructing $\Cs_w$  ensures that this
property holds: this path first goes down till $(S_{k+1}, \tau^\star(S_{k+1}))$, where $S_{k+1}$ is the segment between $v$ and $w$. Subsequently, it moves right till it hits
$\Cs_v$ (see Figure~\ref{fig:opt} for an example). The following crucial lemma shows that is alright to ignore the cells above the path $\Cs_v$. Let $w_1$ and $w_2$ be the children
of $v$ in $\redc(\T)$. We consider the algorithm in Figure~\ref{fig:DP} for filling the DP entry $D[v, \Sst(v)]$. Let $G^\star(v)$ be the edges obtained at the end of Step~3 in this algorithm. Further, let $G^\star_{k+1}(w_r)$ be the set of edges obtained in Step~6(i)(a) of this algorithm when we use the extension $\Sst(w_r)$.

\begin{lemma}
\label{lem:critical}
For $r=1,2$, any path in $\P(v)$ which ends in the segment $S_{k+1}^r$ is satisfied by $G^\star(v) \cup G_{k+1}^\star(w_r)$.
\end{lemma}

\newcommand{\Dom}{{\tt Dom}}
This is the main technical lemma of the contribution and is the key reason why the algorithm works. We will show this by a sequence of steps.
We say that a cell $(S_i, \tau)$ {\em dominates} a cell $(S_j, \tau')$ if $j < i$ and $\tau' - \tau > j-i$.    As in Figure~\ref{fig:opt}, a cell $(S, \tau)$ dominates all cells
  which lie in the upper right quadrant with respect to it if we arrange the cells as shown in the figure.
  For a segment $S_i$,  let $\Dom(S_i)$ be the set of cells dominated by $(S_i, \tau^\star_i)$.
  The following claim shows why this notion is useful. For a set $E$ of edges , let $p(E)$ denote $\sum_{e \in E} p_e$. Recall that $S^\opt(\tau)$ denotes
  the set of edges in $S(\tau)$ selected by the optimal solution.
  \begin{claim}
  \label{cl:dom}
  $$ \sum_{(S_j, \tau) \in \Dom(S_i)} p(S_j^\opt(\tau)) \leq  128^{-\tau^\star_i} B^\star(S_i, \tau^\star_i). $$
  \end{claim}
  \begin{proof}
  Fix a segment $S_j$. For sake of brevity, let $\bar{\tau}$ denote $\tau^\star_i+(i-j)$.
Recall that for any pair $(S, \tau)$, $B^\star(S, \tau) \geq B^\opt(S, \tau)$ (Lemma~\ref{lem:smooth}). 
  Therefore,   terms in the above sum corresponding to $S_j$ add up to 
  $$ \sum_{\tau \geq \bar{\tau}} p(S_j^\opt(\tau)) \leq \sum_{\tau \geq \bar{\tau}} 128^{-\tau} B^\star(S_j, \tau). $$
  By repeated applications of Lemma~\ref{lem:smooth},
  $$B^\star(S_j, \tau)  \leq 8^{\tau-\tau^\star_i}  \cdot 8^{i-j} B^\star(S_i, \tau^\star_i)  . $$
  Therefore,
  \begin{align*} \sum_{\tau \geq \bar{\tau}} p(S_j^\opt(\tau)) & \leq \sum_{\tau \geq \bar{\tau}} 128^{-\tau} \cdot 8^{\tau-\tau^\star_i}  \cdot 8^{i-j} B^\star(S_i, \tau^\star_i) \\
 &  = 128^{-\tau^\star_i}  \sum_{\tau \geq \bar{\tau}} 16^{-(\tau-\tau^\star_i)}   \cdot 8^{i-j} B^\star(S_i, \tau^\star_i) \\
 & \leq 128^{-\tau^\star_i} \cdot \frac{2 \cdot 8^{i-j}}{16^{{\bar{\tau}} - \tau^\star_i}} \cdot B^\star(S_i, \tau^\star_i) = 128^{-\tau^\star_i} \cdot \frac{2}{8^{i-j}}
 \cdot B^\star(S_i, \tau^\star_i)
  \end{align*}
  Summing over all $j < i $ now implies the result.
  \end{proof}

  Let $G^\star_i(v)$ denote the set of edges selected by the Algorithm~{\segmentselect} in Figure~\ref{fig:segmentselect} for the vertex $v$ and segment $S_i$ when called
  with the state $\Sst(v)$. Let $G^\star_i(\tau, v)$ be the density $\tau$ edges in $G^\star_i(v)$.
  The following claim shows that the total size of edges in it is much larger than the corresponding quantity for the optimal solution.
  \begin{claim}
  \label{cl:high}
  For any segment $S_i$,
  $$\sum_{\tau \leq \tau^\star_i} p(G^\star_i(\tau, v))  - \sum_{\tau \leq \tau^\star_i} p(S^\opt_i(\tau)) \geq  128^{-\tau^\star_i} B^\star( S_i, \tau^\star_i).$$
  \end{claim}

  \begin{proof}
Recall that for a segment $S$, density class $\tau$ and budget $B$, $S(\tau, \leq B)$ denotes the edges in $S(\tau)$ which 
have cost at most $B$. The quantity $S(\leq \tau, \leq B)$ was  defined similarly for edges of density class at most $\tau$ in $S$.

  Consider the  Algorithm~{\segmentselect} for $S_i$ with the parameters mentioned above. Note that $\tau^\star_i$ is same as $\tau_1$ in the notation used in Figure~\ref{fig:segmentselect}. Clearly, for $\tau < \tau^\star_i$, the algorithm ensures that $G^\star_i(\tau, v)$ contains $S^\opt_i(\tau)$ (because it selects all edges in $S_i(\tau, \leq B^\star(S_i))$. Since
each egde in $S_i^\opt(\tau)$ has cost at most $B^\opt(S_i, \tau) \leq B^\star(S_i)$, this implies that 
$S_i(\tau, \leq B^\star(S_i))$ contains $S_i^\opt(\tau)$). 
  For the class $\tau_1$, note that the algorithm tries to select edges of total cost at least $4 B^\star(S_i)$. Two cases arise: (i) If it is able to select these many edges,
  then the fact that the optimal solution selects edges of total cost at most $B^\star(S_i)$ from $S_i(\tau_1) $ implies the result, or (ii) The total cost of edges in $S_i(\tau_1, \leq B^\star(S_i))$ is less than $4 B^\star(S_i)$: in this case the algorithm selects all the edges from $S_i(\leq \tau_1, \leq B^\star(S_i))$, and so, $G^\star_i(\tau, v) $
  contains $ S^\opt_i(\tau)$ for all $\tau \leq \tau^\star_i$. The definition of
  $\tau^\star_i$ implies that the total cost of edges in $\cup_{\tau \leq \tau^\star_i} G^\star_i(\tau,v) \setminus S^\opt_i(\tau)$ is at least $4B^\star(S_i)$, and so, the result follows again.
  \end{proof}

  We are now ready to prove Lemma~\ref{lem:critical}. Let $P$ be a path in $\P(v)$ which ends in the segment $S_{k+1}^r$. Suppose $P$ starts in the segment $S_{i_0}$.
  Note that $P$ contains the segments $S_{i_0+1}, \ldots, S_{k}$, but may partially intersect $S_{i_0}$ and $S_{k+1}$.
  \begin{claim}
  \label{cl:one}
  For a cell $(S_i, \tau)$ on the cell sequence $\Cs(v)$, $p(G^\star_i(\tau, v) \cap P) \geq p(S^\opt_i(\tau) \cap P). $
  \end{claim}
  \begin{proof}
  Fix a segment $S_i$ which is intersected by $P$. $P$ contains the lower end-point of this segment $S_i$. If $S^\opt_i(\tau) \subseteq G^\star_i(\tau, v)$, there is nothing
  to prove. Else let $e$ be the first edge in $S^\opt_i(\tau) \setminus G^\star_i(\tau,v)$ as we go up from the lower end-point of this segment. It follows that during Step~5
  of the Algorithm~{\grselect} in Figure~\ref{fig:greedynew}, we would select edges of total cost at least $2B^\opt(S_i,\tau)$
(because $B^\opt(S_i,\tau) \leq B^\star(S_i, \tau)$). The claim follows.
  \end{proof}

  \newcommand{\Above}{{\tt Above}}
  \newcommand{\Below}{{\tt Below}}
  Clearly, if $(S_i, \tau)$ lies below the cell sequence $\Cs(v)$, $G^\star_i(\tau,v)$ contains $S^\opt_i(\tau)$ (because
$G^\star_i(\tau, v)$ is same as $S_i(\tau, \leq B^\star(S_i))$ and $B^\star(S_i) \geq B^\opt(S_i)$). Let $\Above(\Cs(v))$ denotes the cells lying above this sequence,
  and $\Below(\Cs(v))$ the ones lying below it. The above claim now implies that
  \begin{align}
  \label{eq:l1}
\sum_{\tau: (S_{i_0}, \tau) \in \Cs(v) \cup \Below(\Cs(v))} p(S^\opt_{i_0} (\tau) \cap P)
   \leq \sum_{\tau: (S_{i_0}, \tau) \in \Cs(v) \cup \Below(\Cs(v))} p(G^\star_{i_0}(\tau, v) \cap P)
  \end{align}

  Note that any cell $(S_i, \tau)$, $i \geq i_0$, lying above $\Cs(v)$ must be dominated by one of the cells $(S_{i'}, \tau^\star_{i'})$ for $i' = i_0+1, \ldots, k$. Therefore,
  Claim~\ref{cl:dom} shows that
  \begin{align}
 \label{eq:l2}
\sum_{(S_i, \tau)\in \Above(\Cs(v)), i_0 \leq i \leq k} p(S_i^\opt(\tau)) \leq \sum_{i_0 < i \leq k} 128^{-\tau^\star_i} B^\star(S_i, \tau^\star_i)
 \end{align}
  Further, for cells lying on or below $\Cs(v)$, we get using Claim~\ref{cl:high} and Claim~\ref{cl:one}
  \begin{align}
  \label{eq:l3}
  \sum_{(S_i, \tau) \in \Cs(v) \cup \Below(\Cs(v)), i_0 < i \leq k} p(S_i^\opt(\tau)) \leq \sum_{(S_i, \tau) \in \Cs(v) \cup \Below(\Cs(v)), i_0 < i \leq k} G^\star_i(\tau,v)-
 2 \sum_{i_0 < i \leq k} 128^{-\tau^\star_i} B^\star(S_i, \tau^\star_i)
  \end{align}

  Adding the three inequalities above, we see that $\sum_{i=i_0}^k p(S_i^\opt \cap P)$ is at most $\sum_{i=i_0}^k p(G^\star_i(v))$. It remains to consider segment $S_{k+1}^r$.
  In an argument identical to the one in Claim~\ref{cl:one}, we can argue that $p(G_{k+1}^\star(w_r) \cap P) \geq p(S_{k+1}^{r, \opt}\cap P)$, where $S_{k+1}^{r, \opt}$ denotes
  the edges in $S_{k+1}^r$ selected by the optimal solution. This completes the proof of the technical Lemma~\ref{lem:critical}.

  Rest of the task is now easy. We just need to show that DP table entries corresponding to these valid states are comparable to the cost of the optimal solution.
  For a vertex $v \in \redc(\T)$, we shall use the notation $\cells(\T(v))$ to denote the cells $(S, \tau)$, where $S$ lies in the subtree $\T(v)$.
  \begin{lemma}
  \label{lem:costDP}
 For every vertex $v$, the table entry $D[v, \Sst(v)]$ is at most $20 \sum_{(S, \tau) \in \cells(\T(v))} B^\star(S, \tau). $
  \end{lemma}
  \begin{proof}
  We prove by induction on the reduced depth of $v$. If $v$ is a leaf, the lemma follows trivially. Now suppose $v$ has children $w_1$ and $w_2$. Consider the iteration of Step~6 in the algorithm in Figure~\ref{fig:DP}, where we try the extension $\Sst(w_r)$ of the child $w_r$. Lemma~\ref{lem:critical} shows that in Step 6~(b), we will satisfy all paths
  in $\P(v)$ which end in the segment $S_{k+1}^r$. We now bound the cost of edges in $G^r_{k+1}(w_r)$ defined in Step~6(a).
  \begin{claim}
  The cost of $G^r_{k+1}(w_r)$ is at most $20 \sum_{\tau} B^\star(S_{k+1}^r, \tau)$.
  \end{claim}
  \begin{proof}
  We just need to analyze the steps in the algorithm~{\segmentselect} in Figure~\ref{fig:segmentselect}.
  For sake of brevity, let $\tau^\star$ denote $\tau^\star(S_{k+1}^r)$, and $B^\star$ denote $B^\star(S_{k+1}^r)$. The definition of $\tau^\star$ shows that the total cost of edges
  in $S^r_{k+1}(\leq \tau^\star, \leq B^\star)$ is at most $\sum_{\tau \leq \tau^\star} B^\opt(S_{k+1}^r, \tau) + 4B^\star \leq
  \sum_{\tau \leq \tau^\star} B^\star(S_{k+1}^r, \tau) + 4B^\star$. For the density class $\tau^\star$, the set of edges selected would cost at most $6B^\star$,
  because the algorithm in Figure~\ref{fig:greedynew} will take edges of cost up to $3 B^\star$ from either ends. Similarly,
   for
  density classes $\tau$ more than $\tau^\star$, this quantity is at most $6 B^\star(S, \tau)$, Summing up everything, and using the fact that $B^\star  = B^\star(S, \tau)$ for
  some density class $\tau$ gives the result.
  \end{proof}
  The lemma now follows by applying induction on $D[w_r, \Sst(w_r)]$.
  \end{proof}

 Applying the above lemma to the root vertex $r$, we see that $D[r, \emptyset]$ is at most a constant time $\sum_{(S, \tau)} B^\star(S, \tau)$, which by Lemma~\ref{lem:smooth},
 is a constant times the optimal cost. Finally, Lemma~\ref{cor:feas} shows that this entry denotes the cost of a feasible solution. Thus, we have shown the main Theorem~\ref{thm:dmc}.


\section{Discussion} We give the first pseudo-polynomial time constant factor approximation algorithm for the weighted flow-time problem on a single machine. The algorithm can be
made to run in time polynomial in $n$ and $W$ as well, where $W$ is the ratio of the maximum to the minimum weight. The rough idea is as follows. We have already assumed  that
the costs of the job segments are polynomially bounded (this is without loss of generality). Since the cost of a job segment is its weight times its length, it follows that the
lengths of the job segments are also polynomially bounded, say in the range $[1, n^c]$. Now we ignore all jobs of size less than $1/n^2$, and solve the remaining problem using
our algorithm (where
$P$ will be polynomially bounded). Now, we introduce these left out jobs, and show that increase in weighted flow-time will be small.
\\
Further, the algorithm also extends to the problem of minimizing $\ell_p$ norm of weighted flow-times. We can do this by changing the objective function in (IP2) to $(\sum_j \sum_s (w(j,S))^p y(j,S))^{1/p}$ and showing that this is within a constant factor of the optimum value. The instance of ~\dmc in the reduction remains exactly the same, except that the weights of the nodes are now $w(j,S)^p$.
We leave the problem of obtaining a truly polynomial time
constant factor approximation algorithm as open.

\section*{Acknowledgement} The authors would like to thank Sungjin Im for noting the extension to $\ell_p$ norm of weighted flow-times.

\bibliographystyle{plain}
\bibliography{paper}

\end{document}